\documentclass[11pt]{article}
\usepackage[T1]{fontenc}
\usepackage[utf8]{inputenc}
\usepackage[margin=1in]{geometry}
\usepackage{lmodern}
\usepackage{amsmath,amssymb,amsfonts,amsthm}
\usepackage{dsfont,graphicx}
\usepackage{bbm}
\usepackage{xcolor}
\usepackage{algorithm}
\usepackage{algpseudocode}
\graphicspath{{./},{./figures/}}
\usepackage{tikz}
\usepackage[mathscr]{euscript}
\usepackage{url}
\usepackage{preprintstyle}
\usepackage[hidelinks]{hyperref}
\hypersetup{pdftitle={Maximum-Entropy Random Walks on Hypergraphs},
  pdfauthor={Anqi Dong, Anzhi Sheng, Xin Mao, Can Chen}}
\allowdisplaybreaks[1]

\newtheorem{lemma}{Lemma}
\newtheorem{proposition}{Proposition}

\newtheorem{problem}{Problem}

\newtheorem{definition}{Definition}

\newcommand{\bp}{{\mathbf p}}

\newcommand{\bu}{{\mathbf u}}
\newcommand{\bv}{{\mathbf v}}

\newcommand{\black}{\color{black}}
\usepackage{cite}

\newcommand{\support}{{\rm Supp}}

\newcommand{\vsp}{\vspace{0.05in}}

\DeclareMathOperator*{\argmin}{arg\,min}

\begin{document}

\title{Maximum-Entropy Random Walks on Hypergraphs\thanks{This work was supported by the National Science Foundation under Grant CMMI-2611908. \textit{(Corresponding author: Can Chen.)}}}
\author{Anqi Dong\thanks{Anqi Dong is with the Department of Decision and Control Systems, the Department of Mathematics, and Digital Future, KTH Royal Institute of Technology, 100 44 Stockholm, Sweden (e-mail: anqid@kth.se).}
\and Anzhi Sheng\thanks{Anzhi Sheng is with the Department of Decision and Control Systems, KTH Royal Institute of Technology, 100 44 Stockholm, Sweden (e-mail:anzhi@kth.se).} 
\and Xin Mao\thanks{Xin Mao is with the School of Data and Information Sciences, University of North Carolina at Chapel Hill, NC 27599, USA (e-mail: xinm@unc.edu).}
\and Can Chen\thanks{Can Chen is with the School of Data and Information Sciences, the Department of Mathematics, and the Department of Biostatistics, University of North Carolina at Chapel Hill, NC 27599, USA (e-mail:canc@unc.edu).}
}

\date{}
\maketitle

\begin{abstract}
Random walks are fundamental tools for analyzing complex networked systems, including social networks, biological systems, and communication infrastructures. While classical random walks focus on pairwise interactions, many real-world systems exhibit higher-order interactions naturally modeled by hypergraphs. Existing random walk models on hypergraphs often focus on undirected structures or do not incorporate entropy-based inference, limiting their ability to capture directional flows, uncertainty, or information diffusion in complex systems. In this article, we develop a maximum-entropy random walk framework on directed hypergraphs with two interaction mechanisms: broadcasting where a pivot node activates multiple receiver nodes and merging where multiple pivot nodes jointly influence a receiver node. We infer a transition kernel via a Kullback--Leibler divergence projection onto constraints enforcing stochasticity and stationarity. The resulting optimality conditions yield a multiplicative scaling form, implemented using Sinkhorn--Schr\"odinger-type iterations with tensor contractions. We further analyze ergodicity, including projected linear kernels for broadcasting and tensor spectral criteria for polynomial dynamics in merging. The effectiveness of our framework is demonstrated with both synthetic and real-world examples. 
\end{abstract}

\begin{keywords}
\textbf{Keywords:} 
Hypergraphs, maximum-entropy random walks, Markov processes, Sinkhorn--Schr\"odinger algorithm, Kullback--Leibler divergence.
\end{keywords}

\section{Introduction}\label{sec:intro}

Walks on networks turn an adjacency structure into a stochastic dynamics by specifying a transition rule, which defines a Markov chain whose path statistics quantify centrality, mixing, and information flow \cite{newman2018networks,boccaletti2006complex}. This approach is computationally attractive because repeated application of a local operator scales to large-scale systems without requiring a detailed mechanistic model \cite{barabasi2013network}. Random walks on graphs have been extensively studied and shown to be effective in numerous applications, including ranking nodes in social networks, modeling diffusion in biological and ecological systems, identifying communities, and predicting information or epidemic spread \cite{xia2019random}. However, many real-world networked systems exhibit higher-order interactions that involve simultaneous relationships among more than two entities. For instance, in evolutionary dynamics, multiway interactions can shift macroscopic outcomes compared to pairwise models \cite{sheng2024strategy}. Classical graph-based random walk methods, limited to pairwise edges, would fail on such networks, motivating the use of hypergraphs as a natural framework for modeling multiway connectivity.

Hypergraphs generalize graphs by allowing a hyperedge to connect more than two nodes, providing a natural representation for systems with higher-order interactions \cite{berge1984hypergraphs}. Numerous graph-theoretic techniques, including entropy, stability, controllability, similarity,  diffusion processes, and system identification, have been extended to hypergraphs \cite{chen2021controllability,chen2020tensor,surana2022hypergraph,liang2025discrete,cui2025sis,10910206,mao2025identification}. Defining random walks on hypergraphs raises fundamental questions: what constitutes a ``transition'' when a hyperedge represents a multiway interaction, and what node-level dynamics should be tracked? One common approach preserves linearity at the node level by defining a random walk operator or a Laplacian-type diffusion that aggregates multiway incidences and reduces to the standard graph case when hyperedges have size two \cite{carletti2020random,lucas2020multiorder,carletti2020dynamical,hayashi2020hypergraph,chitra2019random}. Another approach treats each step as inherently higher-order, so the next state depends on a multiway context and the node marginals evolve under a nonlinear recursion \cite{benson2017spacey,li2019c}. In this setting, tensor spectral theory plays a central role, including Perron--Frobenius-type results for nonnegative tensors and tensor eigenvalue problems in spectral hypergraph theory \cite{chang2008perron,li2013z,chang2009eigenvalue,pearson2014spectral,raftery1985model}. Despite these advances, existing hypergraph random walk models often lack directional semantics or a principled probabilistic framework for capturing uncertainty in multiway interactions.

Maximum-entropy random walks (MERWs) provide a principled approach to defining stochastic dynamics in complex networks. In the graph setting, rather than assigning outgoing probabilities via local normalization, one selects a path measure that maximizes entropy while remaining supported on the network \cite{dong2024network}. This global criterion can significantly influence localization and mixing properties and has been successfully applied to tasks such as ranking and link prediction \cite{burda2009localization,sinatra2011maximal,li2011link}. Recently, analogous questions have been raised for hypergraphs, including how to define MERWs without reducing hyperedges to pairwise edges \cite{traversa2024unbiased}. From a transport and control perspective, these constructions admit a canonical computational form: a broad class of maximum-entropy scheduling problems can be expressed as Kullback--Leibler (KL) divergence projections under marginal-type constraints, leading to Sinkhorn--Schr\"odinger-type multiplicative scaling algorithms \cite{chen2016robust,dong2025data,chen2017efficient,dongnegative}. In the hypergraph setting, the same mechanism applies, except that the constraints now live on higher-order tensors, so the scaling updates involve tensor contractions followed by rescaling. Nevertheless, extending MERWs to directed hypergraphs remains challenging due to the combinatorial complexity of hyperedge interactions and the nonlinearity of the induced node dynamics.

In this article, we develop a MERW framework on directed hypergraphs without reducing the hypergraph to an ordinary graph. Specifically, we fix a reference transition tensor supported on the directed hyperedges and impose two constraints: stochasticity at the hyperedge level and a prescribed stationary distribution at the node level. Among all admissible transition laws, we select the one that minimizes the KL divergence to the reference. The central modeling step is to define what a directed hyperedge represents as an interaction and to derive the induced evolution of node marginals. We focus on two canonical interaction mechanisms. In broadcasting, a pivot node activates multiple receiver nodes, yielding a closed linear recursion summarized by a projected Markov kernel. In merging, multiple pivot nodes resolve into a single receiver node, producing a nonlinear polynomial map on the simplex. Broadcasting aligns with the standard ``walks on hypergraphs'' perspective in that a local event induces an effective transition on nodes \cite{carletti2020random}, while merging naturally connects to long-term behavior analyses of higher-order stochastic processes and nonlinear Markov operators \cite{benson2017spacey,fasino2020ergodicity,neumann2023nonlinear}. Both mechanisms fit the same KL-projection template, and the optimality conditions admit a multiplicative scaling form, reducing higher-order constraints to a small set of scaling potentials that are updated through repeated tensor contractions and rescaling. We further analyze ergodicity for both types of random walks, examining projected linear kernels for broadcasting and tensor-spectral criteria for the polynomial dynamics in merging.

The remainder of the article is organized as follows. Section~\ref{sec:prelim} introduces tensors, directed hypergraphs, and MERWs on graphs. Section~\ref{sec:broadcast} focuses on broadcasting, covering the uniform and layered models, the KL-projection formulation with scaling iterations, and ergodicity analysis via the projected kernel. Section~\ref{sec:merge} addresses merging, including the corresponding inference problem, scaling iterations, and ergodicity conditions. 
Numerical experiments are presented in Section \ref{sec:num}. Section~\ref{sec:dis} extends the framework to directed hypergraphs in which each hyperedge exhibits a many-to-many interaction mechanism. Finally, Section~\ref{sec:conclusion} concludes with a discussion of future research directions.

\section{Preliminaries}\label{sec:prelim}

\subsection{Tensors}

Tensors are a natural generalization of vectors and matrices to higher dimensions \cite{wei2016theory, chen2024tensor,kolda2009tensor}. The order of a tensor, also referred to as its mode number, corresponds to the number of dimensions it possesses. A $k$th-order tensor can be  represented as $\mathscr{T}\in\mathbb{R}^{n_1\times n_2\times \cdots\times n_k}$, where each $n_j$ denotes the size of the tensor along the $j$th mode. A tensor is called cubical if every mode has the same size, i.e., $n_1=n_2=\cdots=n_k$. A cubical tensor is called symmetric if its entries are invariant under any permutation of the indices, i.e., $\mathscr{T}_{j_1j_2\cdots j_k} = \mathscr{T}_{j_{\sigma(1)}j_{\sigma(2)}\cdots j_{\sigma(k)}}$, where $\sigma$ is a permutation of $\{1,2,\dots,k\}$.  A cubical tensor is called front- or back-symmetric if its entries are invariant under any permutation of the first or last $k-1$ indices.

The mode-$p$ tensor vector multiplication of a tensor $\mathscr{T}\in\mathbb{R}^{n_1\times n_2\times \cdots\times n_k}$ with a vector $\textbf{v}\in\mathbb{R}^{n_p}$  is denoted as $\mathscr{T}\times_p\textbf{v}\in\mathbb{R}^{n_1\times n_2\times\cdots\times n_{p-1}\times n_{p+1}\times \cdots\times n_k}$ whose entries are given by
\begin{equation*}
    (\mathscr{T} \times_p \mathbf{v})_{j_1 \cdots j_{p-1} j_{p+1} \cdots j_k}
=
\sum_{j_p=1}^{n_p} \mathscr{T}_{j_1 j_2 \cdots j_k} \, v_{j_p}.
\end{equation*}
This operation can be naturally extended to all modes, yielding a scalar given by 
\begin{equation*}
    \mathscr{T}\times_1 \textbf{v}_1\times_2\cdots\times_k\textbf{v}_k=\mathscr{T}\times_{1,2,\dots,k}\{\textbf{v}_1,\textbf{v}_2,\dots,\textbf{v}_k\}\in\mathbb{R}
\end{equation*}
for $\textbf{v}_p\in\mathbb{R}^{n_p}$. If $\mathscr{T}$ is symmetric and $\textbf{v}_p=\textbf{v}$ for all $p$, the resulting product is known as the homogeneous polynomial associated with $\mathscr{T}$.

\subsection{Directed Hypergraphs}

Hypergraphs are a natural generalization of graphs that allow edges to connect more than two nodes \cite{berge1984hypergraphs}. 
Formally, a hypergraph is defined as a pair $\mathcal{H} = \{\mathcal{V}, \mathcal{E}\}$, where 
$\mathcal{V} = \{v_1,v_2, \dots, v_n\}$ is a finite set of nodes and 
$\mathcal{E} = \{e_1, e_2,\dots, e_m\}$ is a collection of hyperedges, with each hyperedge 
$e_p \subseteq \mathcal{V}$ for all $p$. If all hyperedges contain exactly $k$
nodes, $\mathcal{H}$ is called a $k$-uniform hypergraph. A directed hypergraph extends this notion by assigning an orientation to each hyperedge. Specifically, a directed hyperedge is an ordered pair $e = (e^-, e^+)$, where 
$e^-, e^+ \subseteq V$ denote the tail and head of the hyperedge, respectively, 
and satisfy $e^- \cap e^+ = \emptyset$. This formulation subsumes ordinary directed graphs as a special case when 
$|e^-| = |e^+| = 1$ for all hyperedges, and provides a flexible framework for modeling 
asymmetric higher-order interactions. For simplicity, we focus on directed hypergraphs in which each hyperedge has either a single tail node and multiple head nodes, or multiple tail nodes and a single head node. This structure naturally arises in applications such as information diffusion, regulatory networks, and group-based contagion processes \cite{gallo1993directed}.

A $k$-uniform directed hypergraph with $n$ nodes and one-tail hyperedges can be represented by a $k$th-order, $n$-dimensional front-symmetric  degree-normalized adjacency tensor $\mathscr{A}\in\mathbb R_+^{n\times n\times\stackrel{k}{\cdots} \times n}$ such that
\begin{equation*}
\mathscr{A}_{j_1j_2\cdots j_k}=
\begin{cases}
\displaystyle \frac{1}{d_{j_1} (k-1)!} \quad \text{if }  (\{v_{j_1}\},\{v_{j_2},v_{j_3},\dots,v_{j_k}\})\in\mathcal{E}\\
0 \qquad \qquad \quad \;\; \text{otherwise}
\end{cases},
\end{equation*}
where $d_{j_1}$ denotes the (outgoing) degree of node $v_{j_1}$ such that
$
d_{j_1} = |\{ e \in \mathcal{E} \text{ }|\text{ } e = (\{v_{j_1}\}, e^+) \} |.
$
With this normalization, for each fixed tail index $j_1$, the adjacency tensor satisfies the row-stochastic constraint, i.e., 
\begin{equation*}
    \sum_{j_2=1}^n\sum_{j_3=1}^n\cdots\sum_{j_k=1}^n \mathscr{A}_{j_1j_2\cdots j_k}=1.
\end{equation*}
The degree-normalized adjacency tensor for a directed hypergraph with one-head hyperedges can be defined analogously, so that the resulting tensor $\mathscr{A}$ is back-symmetric. Finally, non-uniform directed hypergraphs are represented by combining adjacency tensors of different orders. 
Equivalently, a family of tensors 
$\{\mathscr{A}^{(k)}\}_{k=2}^{K}$ is considered, with each tensor corresponding to hyperedges of size $k$ up to a maximum size $K$.

\subsection{MERWs on Graphs}
On a directed graph with $n$ nodes, a time-homogeneous Markov chain \cite{norris1998markov,levin2017markov} is specified by a row-stochastic transition matrix $\textbf{T}  \in \mathbb{R}_+^{n \times n}$ satisfying $\sum_{j=1}^n \textbf{T}_{ij} = 1$ for all $i$. Let $\textbf{p}_t \in \mathbb{R}_+^n$ denote the state distribution at time $t$, and its evolution is given by $\textbf{p}_{t+1}=\textbf{T}^\top\textbf{p}_t$. A stationary distribution $\textbf{p}\in\mathbb{R}^n_{++}$ satisfies $\textbf{p} = \textbf{T}^\top \textbf{p}$. If the chain is irreducible and aperiodic, then the stationary distribution is unique and $\textbf{p}_t \to \textbf{p}$ for any initial condition, with the mixing behavior governed by the spectrum of $\textbf{T}$ \cite{levin2017markov,seneta06}. This convergence property motivates the following standard notion of ergodicity.

\vsp
\begin{definition}[Ergodicity]   
A Markov chain with transition matrix $\textbf{T}$ is called ergodic if it admits a unique stationary distribution $\textbf{p}$ and, for every initial distribution $\textbf{p}_0 \in \mathbb{R}_+^n$ satisfying $\textbf{1}^\top \textbf{p}_0 = 1$ (where \textbf{1} denotes the all-one vector), the sequence $\{\textbf{p}_t\}$ generated by $\textbf{p}_{t+1} = \textbf{T}^\top \textbf{p}_t$ converges to $\textbf{p}$ as $t \to \infty$. Equivalently, $\textbf{T}$ is ergodic if it is primitive, meaning that there exists an integer $p \ge 1$ such that $\textbf{T}^p$ has strictly positive entries.
\end{definition}
\vsp

A common local choice is the unbiased random walk, in which the outgoing degrees at each node are normalized, i.e., the transition matrix is the degree-normalized adjacency matrix. In contrast, MERWs replace local normalization with a global entropy principle on path measures \cite{burda2009localization}. In the formulation adopted here, a nonnegative reference matrix $\textbf{A}$ encodes admissible transitions (e.g., the degree-normalized adjacency matrix), and the transition matrix $\textbf{T}$ is selected by a KL projection under row-stochasticity and stationarity constraints \cite{chen2016robust,chen2017efficient}. Specifically,  the entry-wise KL divergence is defined as
\begin{align*}
D(\textbf{T}_{ij}, \textbf{A}_{ij}) =
\begin{cases}
\displaystyle \textbf{T}_{ij} \log\!\left(\frac{\textbf{T}_{ij}}{\textbf{A}_{ij}}\right) \quad \text{if }\textbf{A}_{ij} > 0,\\
+\infty \quad \text{if }\textbf{A}_{ij} = 0 \text{ and } \textbf{T}_{ij} > 0,
\end{cases}
\end{align*}
with the convention $0 \log(0/0) = 0$, and solve
\begin{align*}
\argmin_{\textbf{T}} \sum_{i=1}^n\sum_{j=1}^n D(\textbf{T}_{ij}, \textbf{A}_{ij}),
\end{align*}
subject to the constraints that each row of $\textbf{T}$ sums to 1 and that the distribution $\textbf{p}$ is stationary under $\textbf{T}$, i.e., $\textbf{T}^\top \textbf{p} = \textbf{p}$. The resulting optimality conditions have a multiplicative scaling form and can be computed efficiently using Sinkhorn-type iterations \cite{sinkhorn1967concerning,peyre2019computational,chen2016robust}. 

In the following, we generalize the concept of MERWs to directed hypergraphs, which involve two distinct interaction mechanisms that induce two types of node-marginal dynamics: a linear Markov recursion corresponding to broadcasting interactions, and a nonlinear polynomial recursion corresponding to merging interactions.
\begin{figure}[t]
    \centering
    \includegraphics[width=0.6\linewidth]{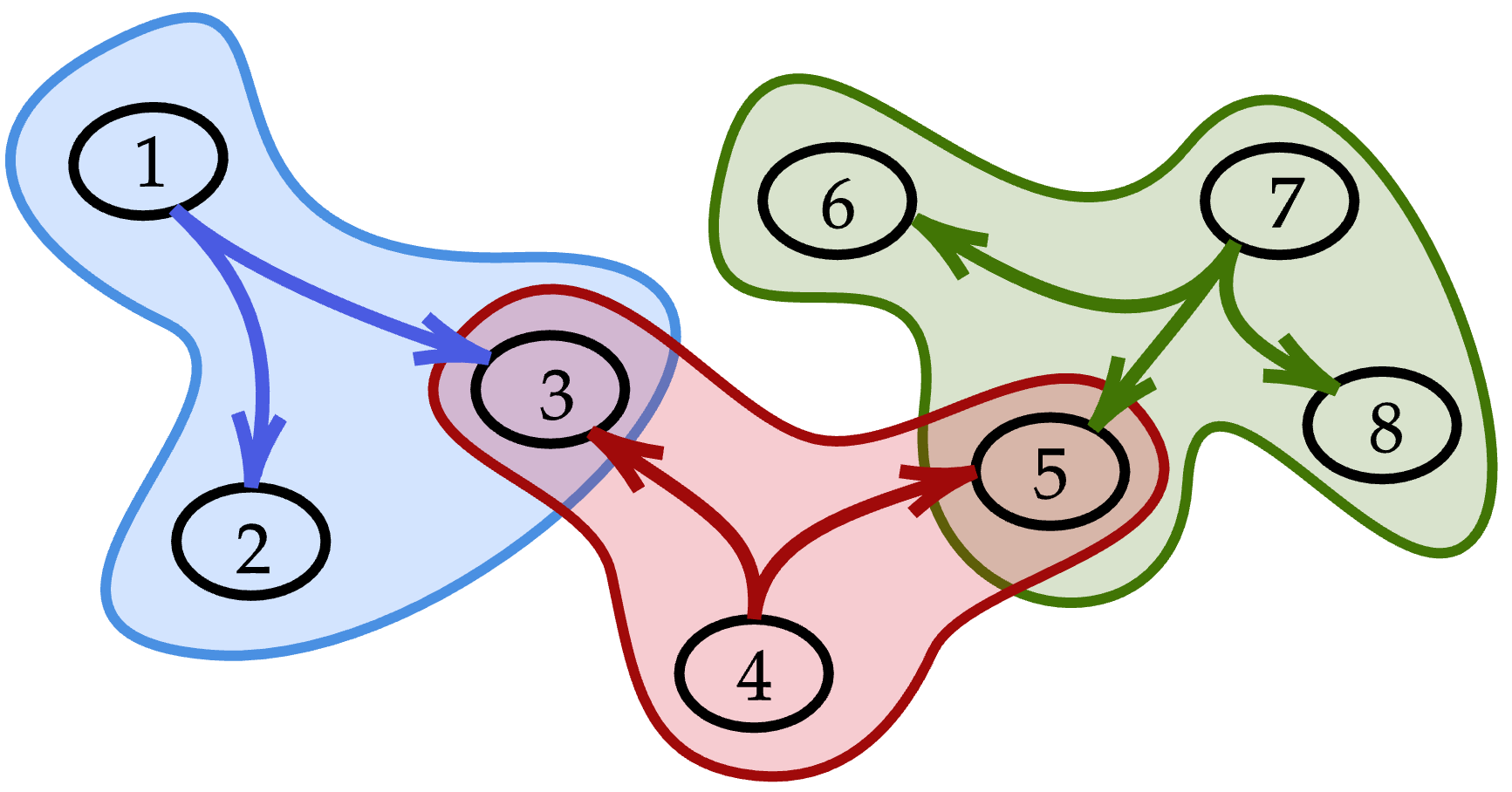}
    \caption{Broadcasting MERWs on directed hypergraphs. Each colored region represents a directed hyperedge, with arrows indicating activation from a pivot (tail) node to a set of receiver (head) nodes (e.g., $v_1\rightarrow \{v_2,v_3\}$).}
    \label{fig:1}
\end{figure}

\section{MERWs for Higher-Order Broadcasting}\label{sec:broadcast}
We first consider the broadcasting setting, in which each directed hyperedge represents a one-to-many transition. This formulation is standard in the random-walk literature: the Markov chain is defined at the node level, and transitions along hyperedges are obtained by projecting each set-valued event onto pairwise node transitions, e.g.,~\cite{carletti2020random,carletti2021random,mulas2022random}. Such higher-order random walks naturally describe processes in social learning and evolutionary game theory~\cite{sheng2024strategy,alvarez2021evolutionary,wang2025emergence,wang2026strategy}, corresponding to simple contagion dynamics under neutral drift, where individuals engage in multi-player interactions, e.g., public goods games~\cite{santos2008social}, and update their states by imitating neighbors. Each imitation interaction in a hyperedge can thus be decomposed into multiple pairwise interactions between a focal individual and several learning individuals.  Here, we aim to retain an explicit broadcasting model and infer its transition law via a maximum-entropy principle.
{ More specifically, the higher-order object being inferred is the conditional law of a complete receiver set given the pivot. The projected node-level kernel retains only a one-receiver marginal, so distinct event-level laws may induce the same linear Markov chain while assigning different probabilities to joint receiver sets. Thus, although the induced node-level dynamics are linear, the optimization retains the higher-order event structure. In contrast to approaches that operate on a projected node-level kernel or an expanded node--hyperedge state space, our method performs the KL projection directly on the directed event tensor subject to its original hypergraph support, stochasticity, and prescribed-stationarity constraints. The node-level Markov kernel is derived only afterward by marginalizing the optimized event-level transition law.}

\subsection{Broadcasting MERWs on Uniform Hypergraphs}

Consider a $k$-uniform directed hypergraph with $n$ nodes, in which each hyperedge has a single tail node and $k-1$ head nodes. The broadcasting transition tensor is a nonnegative, $k$th-order, $n$-dimensional front-symmetric  tensor $\mathscr{B}\in\mathbb R_+^{n\times n\times\stackrel{k}{\cdots} \times n}$ such that for each pivot (tail) node $v_{j_1}$, the entries $\mathscr{B}_{j_1 j_2 \cdots j_k}$ encode the probability of activating the receiver (head) node set
$\{v_{j_2},v_{j_3},\dots,v_{j_k}\}$. This definition is analogous to the degree-normalized adjacency tensor, but unlike the latter, the transition probabilities are not determined by node degrees.

\begin{definition}[Broadcasting]
Let $\mathscr{B}\in\mathbb R_+^{n\times n\times\stackrel{k}{\cdots} \times n}$ be a broadcasting transition tensor and  $\textbf{p}_t \in \mathbb{R}_+^n$ be the state distribution  at time $t$. The broadcasting dynamics, which induce the node-level recursion, is defined as
\begin{equation}
\textbf{p}_{t+1} = \mathscr{B} \times_1 \textbf{p}_t \times_{3,4,\dots,k} \{\textbf{1},\textbf{1},\stackrel{k-2}{\dots},\textbf{1}\},
\end{equation}
where the mode-1 multiplication propagates the probabilities of the pivot node, and the subsequent multiplication over modes $3$ through $k$ sums over all possible choices of the remaining $k-2$ receiver nodes in each hyperedge.
\end{definition}

The broadcasting dynamics remain linear, even though the underlying interactions are higher-order and set-valued.  To define admissible broadcasting kernels, we impose two key requirements. First, each pivot must distribute one unit of probability mass over the set-valued receivers, enforced via the row-stochastic constraint. Second, the induced node-level recursion must admit a prescribed stationary distribution, ensuring consistency between the hyperedge-level dynamics and the node-level behavior. Specifically, let $\mathcal{X}_+$ denote the cone of nonnegative, $k$th-order, $n$-dimensional front-symmetric tensors. Fix $\textbf{p} \in \mathbb{R}^n_{++}$ with $\textbf{1}^\top \textbf{p} = 1$. We then define the admissible class of broadcasting kernels as
\begin{equation}\label{eq:constraint}
\begin{split}
\mathcal B(\textbf{p})
=
\Big\{\mathscr{B} \in \mathcal X_+ \ \mid \ \mathscr{B}  \times_{2,3,\dots,k}\{\textbf{1},\textbf{1},\stackrel{k-1}{\dots},\textbf{1}\} = \textbf{1},
\mathscr{B}  \times_{1}\textbf{p} \times_{3,4,\dots,k}\{\textbf{1},\textbf{1},\stackrel{k-2}{\dots},\textbf{1}\} = \textbf{p}
\Big\}.
\end{split}
\end{equation}

We now formulate maximum-entropy inference for the broadcasting kernel. Given a hypergraph support tensor $\mathscr{A}$ (e.g., the degree-normalized adjacency tensor) and a target stationary distribution $\textbf{p}$, we select the broadcasting transition tensor $\mathscr{B}$ that minimizes the relative entropy with respect to $\mathscr{A}$ among all kernels satisfying the constraints in \eqref{eq:constraint}.

\begin{problem} \label{prob:1}
Given a $k$-uniform directed hypergraph with a single tail per hyperedge and $n$ nodes, fix a strictly positive distribution $\textbf{p} \in \mathbb{R}^n_{++}$ with $\textbf{1}^\top \textbf{p} = 1$, and let $\mathscr{A} \in \mathbb{R}_+^{n \times n \times \stackrel{k}{\cdots} \times n}$ denote the degree-normalized adjacency tensor. The corresponding broadcasting transition tensor $\mathscr{B} \in \mathbb{R}_+^{n \times n \times \stackrel{k}{\cdots} \times n}$ is obtained by solving
\begin{align}\label{eq:obj1}
\min_{\mathscr{B}\in \mathcal B(\textbf{p})}\ 
\sum_{j_1=1}^n\sum_{j_2=1}^n\cdots\sum_{j_k=1}^n \mathscr{B}_{j_1j_2\cdots j_k}
\log\!\Big(\frac{\mathscr{B}_{j_1j_2\cdots j_k}}{\mathscr{A}_{j_1 j_2\cdots j_k}}\Big),
\end{align}
subject to the support constraint $\support(\mathscr{B}) \subseteq \support(\mathscr{A})$.\footnote{$\support(\mathscr{B}) \subseteq \support(\mathscr{A})$ implies that $\mathscr{B}_{j_1 j_2 \cdots j_k}$ can be positive only if $\mathscr{A}_{j_1 j_2 \cdots j_k}$ is also positive. Equivalently, whenever $\mathscr{A}_{j_1 j_2 \cdots j_k}=0$, we must have $\mathscr{B}_{j_1 j_2 \cdots j_k}=0$, with the convention $\log(0/0)=0$.}
\end{problem}

We solve \eqref{eq:obj1} by optimizing directly over the broadcasting transition tensor $\mathscr{B}$. Since both constraints in \eqref{eq:constraint} are expressed through tensor contractions, the resulting KL projection admits a multiplicative scaling structure analogous to that of classical matrix scaling. To make this multiplicative structure explicit, we introduce a pivot-weighted reference tensor supported on the hyperedges defined as 
\begin{equation*}\label{eq:prior-weighted-uniform}
\mathscr{K}_{j_1j_2\cdots j_k}=\textbf{p}_{j_1}\mathscr{A}_{j_1j_2\cdots j_k},
\quad
\support(\mathscr{K})=\support(\mathscr{A}).
\end{equation*}
This reweighting does not alter the optimizer of \eqref{eq:obj1}. Indeed, since the row-stochasticity constraint  implies that, for each pivot index $j_1$, the probabilities over all receiver sets sum to one, we obtain
\begin{align*}
\sum_{j_1=1}^n\sum_{j_2=1}^n\cdots\sum_{j_k=1}^n\mathscr{B}_{j_1j_2\cdots j_k}\log\!\Big(\frac{\mathscr{B}_{j_1j_2\cdots j_k}}{\mathscr{A}_{j_1j_2\cdots j_k}}\Big)
=
\sum_{j_1=1}^n\sum_{j_2=1}^n\cdots\sum_{j_k=1}^n\mathscr{B}_{j_1j_2\cdots j_k}\log\!\Big(\frac{\mathscr{B}_{j_1j_2\cdots j_k}}{\mathscr{K}_{j_1j_2\cdots j_k}}\Big)+\sum_{j_1=1}^n\log \textbf{p}_{j_1},
\end{align*}
where the final term is constant over $\mathcal{B}(\textbf{p})$. Consequently, the KL projection onto $\mathcal{B}(\textbf{p})$ is identical whether the reference tensor is $\mathscr{A}$ or $\mathscr{K}$.

\begin{proposition}\label{prop:1}
Assume that the constraint set $\mathcal B(\textbf{p})$ is nonempty. Problem~\ref{prob:1} admits a unique optimal broadcasting transition tensor $\mathscr{B}^\star \in \mathcal B(\textbf{p})$. Moreover, there exist strictly positive vectors $\textbf{u}, \textbf{v} \in \mathbb R^n_{++}$ such that $\mathscr{B}^\star$ admits the multiplicative factorization 
\begin{equation}\label{eq:close-form1}
\mathscr{B}^\star = 
\mathscr{K}\odot\big(\textbf{u}\circ (\textbf{v}\circ\textbf{v}\circ\stackrel{k-1}{\cdots}\circ \textbf{v})\big),
\end{equation}
where $\odot$ and $\circ$ denote entry-wise multiplication and vector outer product, respectively.
\end{proposition}

\begin{proof}
Since the admission set $\mathcal B(\textbf{p})$ is nonempty, it is convex and closed, as it is defined by linear equalities together with the nonnegativity constraint $\mathscr{B} \ge 0$ and the support restriction $\support(\mathscr{B}) \subseteq \support(\mathscr{A})$. Moreover, the row-stochasticity constraint ensures that $\mathcal B(\textbf{p})$ is bounded, and hence compact. The objective in \eqref{eq:obj1} is lower semicontinuous on $\mathbb{R}_+^{n\times n\times \stackrel{k}{\cdots}\times n}$ and finite on $\mathcal B(\textbf{p})$ under the support restriction, which guarantees the existence of a minimizer. Furthermore, since the function $x \mapsto x \log(x/c)$ is strictly convex on $x \ge 0$ for any $c>0$, the objective is strictly convex over $\support(\mathscr{A})$, implying that the minimizer over the convex set $\mathcal B(\textbf{p})$ is unique.

Recall that replacing $\mathscr{A}$ by the pivot-weighted prior $\mathscr{K}$ changes the objective in \eqref{eq:obj1} only by an additive constant over $\mathcal B(\textbf{p})$, and thus does not affect the optimizer. Therefore, we equivalently minimize
\begin{equation*}
\sum_{j_1=1}^n\sum_{j_2=1}^n\cdots\sum_{j_k=1}^n\mathscr{B}_{j_1j_2\cdots j_k}
\log\!\Big(\frac{\mathscr{B}_{j_1j_2\cdots j_k}}{\mathscr{K}_{j_1j_2\cdots j_k}}\Big)
\end{equation*}
over $\mathscr{B}\in\mathcal B(\textbf{p})$. Introduce Lagrange multipliers $\boldsymbol{\alpha} \in \mathbb{R}^n$ and $\boldsymbol{\beta} \in \mathbb{R}^n$ corresponding to the two constraints in \eqref{eq:constraint}, and consider the Lagrangian
\begin{align*}
\mathcal L(\mathscr{B},\boldsymbol{\alpha},\boldsymbol{\beta})
&=
\sum_{j_1=1}^n\sum_{j_2=1}^n\cdots\sum_{j_k=1}^n\mathscr{B}_{j_1j_2\cdots j_k}
\log\!\Big(\frac{\mathscr{B}_{j_1j_2\cdots j_k}}{\mathscr{K}_{j_1j_2\cdots j_k}}\Big)+\sum_{j_1=1}^n\boldsymbol{\alpha}_{j_1}\Big(\sum_{j_2=1}^n\cdots\sum_{j_k=1}^n\mathscr{B}_{j_1j_2\cdots j_k}-1\Big)\\
&+\sum_{j_2=1}^n\boldsymbol{\beta}_{j_2}\Big(\sum_{j_1=1}^n\sum_{j_3=1}^n\cdots\sum_{j_k=1}^n\textbf{p}_{j_1}\mathscr{B}_{j_1 j_2\cdots j_k}-\textbf{p}_{j_2}\Big).
\end{align*}
At the optimum, for any index $(j_1,j_2, \dots, j_k) \in \support(\mathscr{K})$, the first-order condition gives
\begin{equation*} 
\frac{\partial \mathcal L}{\partial \mathscr{B}_{j_1j_2\cdots j_k}}
=
\log\!\Big(\frac{\mathscr{B}^\star_{j_1j_2\cdots j_k}}{\mathscr{K}_{j_1j_2\cdots j_k}}\Big)+1
+\boldsymbol{\alpha}_{j_1}
+\textbf{p}_{j_1}\,\boldsymbol{\beta}_{j_2}
=0.
\end{equation*}
Since $\mathscr{B}^\star$ is front-symmetric, exponentiating both sides yields the multiplicative form
\begin{equation*}
\mathscr{B}^\star_{j_1j_2\cdots j_k}
=
\mathscr{K}_{j_1j_2\cdots j_k}\textbf{u}_{j_1}\textbf{v}_{j_2}\textbf{v}_{j_3}\cdots \textbf{v}_{j_k}
\end{equation*}
where $\textbf{u}_{j}=\exp(-1-\boldsymbol{\alpha}_{j})$ and $\textbf{v}_{j}=\exp(-\textbf{p}_{j}\boldsymbol{\beta}_{j})$. Therefore, the result follows immediately. 
\end{proof}

The result ensures that the optimal maximum-entropy broadcasting transition tensor $\mathscr{B}^\star$ is unique and emphasizes that the optimizer preserves the hypergraph topology encoded in $\mathscr{K}$ while adjusting probabilities to satisfy both row-stochasticity and the prescribed node-level stationary distribution. Importantly, its multiplicative, front-symmetric structure provides a natural foundation for efficient iterative algorithms based on alternating rescaling, directly analogous to matrix scaling for KL projections~\cite{chen2016robust}.

\subsection{Extension to Non-Uniform Hypergraphs}

The broadcasting MERWs can be naturally extended to general non-uniform hypergraphs by decomposing them into layers of different uniformities. For each  $k \in \{2,3,\dots,K\}$, let $\mathscr{A}^{(k)} \in \mathbb{R}_+^{n \times n\times \stackrel{k}{\cdots} \times n}$ denote the degree-normalized adjacency tensors of $k$-uniform hyperedges, and let $\mathscr{B}^{(k)} \in \mathbb{R}_+^{n \times n\times \stackrel{k}{\cdots} \times n}$ be the corresponding broadcasting transition tensors. We adopt the same storage convention as in the uniform case, i.e.,  $\mathscr{B}^{(k)}$ are front-symmetric, and given a pivot index $j_1$, the entries $\mathscr{B}^{(k)}_{j_1 j_2\cdots j_k}$ are indexed by an ordered receiver tuple $(j_2,j_3,\dots,j_k)$. To account for the relative importance of different hyperedge sizes in the overall dynamics, we introduce nonnegative layer weights $\{\lambda_k\}_{k=2}^{K}$ with summation equal to one. This layered decomposition preserves the higher-order structure of the original hypergraph while enabling the extension of the maximum-entropy broadcasting framework to heterogeneous hypergraphs with mixed arities.

Let $\textbf{p}_t\in\mathbb R_+^n$ denote the state distribution at time $t$.
Using the same projection rule as in the uniform case, the induced node-level broadcasting dynamics on a non-uniform hypergraph is given by the weighted superposition
\begin{align}\label{eq:nonuniform-dynamics}
\textbf{p}_{t+1}
=
\sum_{k=2}^{K}
\lambda_k\,
\mathscr{B}^{(k)}\times_{1}\textbf{p}_t\times_{3,4,\dots,k}\{\textbf{1},\textbf{1},\stackrel{k-2}{\dots},\textbf{1}\}.
\end{align}
The induced evolution of node marginals remains linear in $\textbf{p}_t$. To define the admissible class of non-uniform broadcasting kernels, let $\mathcal X^{(k)}_+$ denote the cone of nonnegative $k$th-order front-symmetric tensors, and
fix $\textbf{p}\in\mathbb R^n_{++}$ with $\textbf{1}^\top\textbf{p}=1$. The admissible set is then defined as
\begin{equation}\label{eq:constraint2}
\begin{split}
\mathcal B_{\le K}(\bp)
=
\Big\{\{
\mathscr{B}^{(k)}\}_{k=2}^K \ \mid \ \mathscr{B}^{(k)}\in\mathcal X^{(k)}_+,
&\sum_{k=2}^{K}
\lambda_k\,
\mathscr{B}^{(k)}\times_{2,3,\dots,k}\{\textbf{1},\textbf{1},\stackrel{k-1}{\dots},\textbf{1}\}=
\textbf{1},\\
&\sum_{k=2}^{K}
\lambda_k\,
\mathscr{B}^{(k)}\times_{1}\textbf{p}\times_{3,4,\dots,k}\{\textbf{1},\textbf{1},\stackrel{k-2}{\dots},\textbf{1}\}
=
\textbf{p}
\Big\}. 
\end{split}
\end{equation}
Non-uniform broadcasting kernels are inferred via the same maximum-entropy principle as in the uniform case, now with a layer-weighted KL divergence objective that accounts for the relative contribution of each uniform layer.

\begin{problem} \label{prob:nonuniform-prime}
Given a directed hypergraph with a single tail per hyperedge, $n$ nodes, and maximum hyperedge size $K$, fix a strictly positive distribution $\textbf{p} \in \mathbb{R}^n_{++}$ with $\textbf{1}^\top \textbf{p} = 1$, and let $\mathscr{A}^{(k)} \in \mathbb{R}_+^{n \times n \times \stackrel{k}{\cdots} \times n}$ denote the degree-normalized adjacency tensors for $k$-uniform layers. The corresponding broadcasting transition tensors $\mathscr{B}^{(k)} \in \mathbb{R}_+^{n \times n \times \stackrel{k}{\cdots} \times n}$ are obtained by solving
\begin{equation}\label{eq:obj2}
\displaystyle\min_{\{\mathscr{B}^{(k)}\}_{k=2}^K\in \mathcal B_{\leq K}(\textbf{p})}\sum_{k=2}^K\lambda_k\ 
\sum_{j_1=1}^n\sum_{j_2=1}^n\cdots\sum_{j_k=1}^n \mathscr{B}^{(k)}_{j_1j_2\cdots j_k}
\log\!\Big(\frac{\mathscr{B}^{(k)}_{j_1j_2\cdots j_k}}{\mathscr{A}^{(k)}_{j_1 j_2\cdots j_k}}\Big),
\end{equation}
subject to the support constraint $\operatorname{supp}(\mathscr{B}^{(k)}) \subseteq \operatorname{supp}(\mathscr{A}^{(k)})$ for $k=2,3,\dots,K$.
\end{problem}

\begin{proposition}\label{prop:pro2}
Assume that the constraint set $\mathcal B_{\le K}(\textbf{p})$ is nonempty. Problem~\ref{prob:nonuniform-prime} admits a unique set of optimal broadcasting transition tensors $\{\mathscr{B}^{(k)\star}\}_{k=2}^K \in \mathcal B_{\le K}(\textbf{p})$. Moreover, there exist strictly positive vectors $\textbf{u}, \textbf{v} \in \mathbb R^n_{++}$ such that $\mathscr{B}^{(k)\star}$ admits the multiplicative factorization 
\begin{equation}\label{eq:close-form2}
\mathscr{B}^{(k)\star} = 
\mathscr{K}^{(k)}\odot\big(\textbf{u}\circ (\textbf{v}\circ\textbf{v}\circ\stackrel{k-1}{\cdots}\circ \textbf{v})\big),
\end{equation}
where $\mathscr{K}^{(k)}$ denote the pivot-weighted reference tensors for $k$-uniform layers, for $k=2,3,\dots,K$. 
\end{proposition}
\begin{proof}
Existence of a solution follows from feasibility and compactness of $\mathcal B_{\le K}(\textbf{p})$, and uniqueness from strict convexity of the KL objective on the prescribed support. The multiplicative form of the optimizer can be derived by introducing Lagrange multipliers $\boldsymbol{\alpha} \in \mathbb{R}^n$ and $\boldsymbol{\beta} \in \mathbb{R}^n$ associated with the two constraints in \eqref{eq:constraint2}, and forming the Lagrangian corresponding to \eqref{eq:obj2}. The first-order optimality condition yields
\begin{equation*}
\lambda_k\Big(\log\!\Big(\frac{\mathscr{B}^{(k)\star}_{j_1j_2\cdots j_k}}{\mathscr{K}^{(k)}_{j_1j_2\cdots j_k}}\Big)+1\Big)
+\lambda_k\boldsymbol{\alpha}_{j_1}
+\lambda_k\textbf{p}_{j_1}\,\boldsymbol{\beta}_{j_2}
=0.
\end{equation*}
Since $\mathscr{B}^{(k)\star}$ are front-symmetric, exponentiating both sides yields the stated multiplicative form for some positive vectors $\textbf{u}$ and $ \textbf{v}$ for $k=2,3,\dots,K$.
\end{proof}
\vsp 

This formulation can be interpreted as a one-step transport problem on a directed hypergraph. Rather than enforcing stationarity, the broadcasting tensors are chosen so that a single broadcasting operation pushes an initial distribution $\textbf{q}\in\mathbb{R}^{n}_{++}$ with $\textbf{1}^\top\textbf{q}=1$ to a prescribed target distribution $\textbf{p}$, aggregated over all hyperedge sizes. The endpoint constraint enforces this mass transfer, while the row-stochastic constraint ensures that each $\mathscr{B}^{(k)}$ defines a valid broadcasting kernel, yielding an entropy-regularized one-step transport guided by the hypergraph structure.

\subsection{Ergodicity}
A key simplification of the broadcasting model is that, although the underlying dynamics arise from higher-order, set-valued interactions, the induced evolution of the node marginals is governed by an ordinary linear time-invariant system on $\mathbb{R}^n$. This reduction makes ergodicity tractable, as mixing and convergence can be analyzed through linear operators rather than the exponentially large tensor space. At the same time, the inference problem remains global: the constraints in \eqref{eq:constraint2} enforce mixture mass conservation and stationarity with respect to the prescribed distribution \textbf{p} by aggregating information across all hyperedge sizes and layers.

Based on the induced node-level non-uniform broadcasting dynamics \eqref{eq:nonuniform-dynamics}, define the projected transition kernel
\begin{equation*}
    \textbf{P} = \sum_{k=2}^{K}
\lambda_k\,
\mathscr{B}^{(k)}\times_{3,4,\dots,k}\{\textbf{1},\textbf{1},\stackrel{k-2}{\dots},\textbf{1}\}\in\mathbb{R}^{n\times n}_{+}.
\end{equation*}
With this definition, the broadcasting dynamics reduce to the linear system $\textbf{p}_{t+1} = \textbf{P}^\top\textbf{p}_{t}$.

\begin{lemma}
    If $\{\mathscr{B}^{(k)}\}_{k=2}^K\in\mathcal B_{\le  K}(\textbf{p})$, the projected transition kernel is row-stochastic and admits \textbf{p} as a stationary distribution, i.e., $\textbf{P}\textbf{1}=\textbf{1}$ and $\textbf{P}^\top\textbf{p}=\textbf{p} $.
\end{lemma}
\begin{proof}
The results follow immediately from the definition of the projected kernel \textbf{P} and the properties of tensor vector multiplication, as well as the row-stochasticity and stationarity constraints defined in $\mathcal B_{\le  K}(\textbf{p})$.
\end{proof}

The lemma shows that the projected kernel \textbf{P} preserves the two key structural properties of the linear system: conservation of total mass (row-stochasticity) and invariance of the prescribed equilibrium $\textbf{p}$. As a result, ergodicity and mixing behavior  (how fast this convergence occurs) of broadcasting MERWs can be analyzed through the spectral properties of $\textbf{P}^\top$ using standard finite-state Markov chain theory.
In particular, once $\textbf{P}$ is formed, convergence to $\textbf{p}$, mixing rates, and spectral gaps can be studied entirely at the node level, without ever manipulating the full family of high-order tensors $\{\mathscr{B}^{(k)}\}_{k=2}^K$.

\subsection{Sinkhorn--Schr\"odinger Scaling}
Since broadcasting MERWs on general non-uniform directed hypergraphs can be formulated as a KL projection onto two linear constraint sets, the resulting inference problem has a simple and structured solution. In particular, the constraints enforce global mass conservation of the event mixture and stationarity of the induced node-level dynamics, while the KL objective selects the closest broadcasting model to the reference hypergraph. This variational structure leads to a multiplicative scaling form for the optimal broadcasting transition tensors $\mathscr{B}^{(k)}$, which can be computed efficiently via a Sinkhorn--Schr\"odinger type iteration.

For strictly positive scaling vectors $\textbf{u},\textbf{v}\in\mathbb R^n_{++}$, define the scaled broadcasting transition tensors
\begin{equation*}
\mathscr{B}^{(k)}(\textbf{u},\textbf{v})
=
\mathscr{K}^{(k)} \odot\big(\textbf{u}\circ (\textbf{v}\circ\textbf{v}\circ\stackrel{k-1}{\cdots}\circ \textbf{v})\big),
\end{equation*}
where $\mathscr{K}^{(k)}_{j_1j_2\cdots j_k} =\textbf{p}_{j_1} \mathscr{A}^{(k)}_{j_1j_2\cdots j_k}$ denote the pivot-weighted reference tensors for the $k$-uniform layers. Here, the scaling vector \textbf{u} is associated with the pivot mode, while the vector \textbf{v} is shared across all receiver modes. Under this parametrization, the constraints in \eqref{eq:constraint2} can be expressed in terms of the induced mixture marginals, i.e., 
\begin{equation}\label{eq:marginals}
    \begin{split}
        \varphi(\textbf{u},\textbf{v})&=\mathscr{B}^{(k)}(\textbf{u},\textbf{v})\times_{2,3,\dots,k}\{\textbf{1},\textbf{1},\stackrel{k-1}{\dots},\textbf{1}\},\\
\eta(\textbf{u},\textbf{v})&=\sum_{k=2}^{K}
\lambda_k\,
\mathscr{B}^{(k)}(\textbf{u},\textbf{v})\times_{1}\textbf{p}\times_{3,4,\dots,k}\{\textbf{1},\textbf{1},\stackrel{k-2}{\dots},\textbf{1}\}.
\end{split}
\end{equation}
The feasibility conditions in \eqref{eq:constraint2} are therefore equivalent to $\varphi(\textbf{u},\textbf{v})= \textbf{1}$ and $\eta(\textbf{u},\textbf{v})=\textbf{p}$. Each marginal is affine in one scaling vector when the other is held fixed. This structure naturally leads to an alternating rescaling procedure of Sinkhorn--Schr\"odinger type:
\begin{align*}
\textbf{u} \gets \textbf{u}\odot\big(\textbf{1} \oslash \varphi(\textbf{u},\textbf{v})\big),\quad
\textbf{v} \gets \textbf{v}\odot\big(\textbf{p} \oslash \eta(\textbf{u},\textbf{v})\big),    
\end{align*}
where $\oslash$ denotes entry-wise division. 
These updates iteratively enforce the two constraints by normalizing the pivot and receiver marginals in turn. Upon convergence, the resulting tensors $\mathscr{B}^{(k)\star}=\mathscr{B}^{(k)}(\bv^\star,\bu^\star)$ for all $k=2,3,\dots,K$ constitute the solution to the KL projection problem. Algorithm~\ref{alg:sinkhorn-broadcast-nonuniform} summarizes this procedure and serves as a direct analogue of classical two-vector Sinkhorn scaling. In particular, the 
\textbf{u}-update enforces event-level mass conservation of the mixture, while the 
\textbf{v}-update enforces stationarity of the induced node-marginal dynamics.

\begin{algorithm}[t]
\caption{Sinkhorn--Schr\"odinger Scaling for Broadcasting}\label{alg:sinkhorn-broadcast-nonuniform}
\begin{algorithmic}[1]
\Require State distribution $\textbf{p}\in\mathbb R^n_{++}$ with $\textbf{1}^\top\textbf{p}=1$, adjacency tensors $\{\mathscr{A}^{(k)}\}_{k=2}^{ K}$,
layer weights $\{\lambda_k\}_{k=2}^{ K}$, tolerance $\varepsilon>0$
\State Initialize $\textbf{u}\gets \textbf{1}$, $\textbf{v}\gets \textbf{1}$
\While{ not converge }
\State Compute $\varphi(\textbf{u},\textbf{v})$ based on \eqref{eq:marginals}
\State Set $\textbf{u} \gets \textbf{u}\odot\big(\textbf{1} \oslash \varphi(\textbf{u},\textbf{v})\big)$
\State Compute $\eta(\textbf{u},\textbf{v})$ based on \eqref{eq:marginals}
\State Set $\textbf{v} \gets \textbf{v}\odot\big(\textbf{p} \oslash \eta(\textbf{u},\textbf{v})\big)$
\EndWhile
\State \textbf{Return:} $\mathscr{B}^{(k)\star}\gets \mathscr{B}^{(k)}(\textbf{u},\textbf{v})$ for $k=2,3,\dots, K$
\end{algorithmic}
\end{algorithm}

\begin{proposition}[Convergence]
Under the regularity conditions for dual coordinate ascent, Algorithm~\ref{alg:sinkhorn-broadcast-nonuniform} admits a fixed point $(\textbf{u}^{\star},\textbf{v}^{\star})$ corresponding to the unique optimal broadcasting transition tensors, and the induced transition tensors converge to the optimal solution at a linear rate.
\end{proposition}

\begin{proof}
The constraints in \eqref{eq:constraint2} prescribe the
pivot and receiver marginals.
Algorithm~\ref{alg:sinkhorn-broadcast-nonuniform} alternates the two
multiplicative updates
\begin{equation*}\label{eq:app-v-update}
\textbf{u}^{+}
=
\textbf{u}\odot
\big(\textbf{1}\oslash\varphi(\textbf{u},\textbf{v})\big),
\end{equation*}
which enforces
$\varphi(\textbf{u}^{+},\textbf{v})=\textbf{1}$, and
\begin{equation*}\label{eq:app-u-update}
\textbf{v}^{+}
=
\textbf{v}\odot
\big(\textbf{p}\oslash\eta(\textbf{u},\textbf{v})\big).
\end{equation*}
Each update corresponds to an
exact block-wise maximization of the dual in the log-potentials.
Strict convexity of the primal KL objective on the prescribed support
guarantees a unique optimizer
$\{\mathscr{B}^{(k)\star}\}_{k=2}^{K}$.
Moreover, the dual objective is smooth and
concave in the log-potentials, and the alternating updates
implement block coordinate ascent. Therefore, by the linear
convergence result in \cite{luo1992convergence}, the dual iterates approach the optimal dual solution set
at a linear rate. Since the transition tensors depend smoothly on the logarithmic dual variables, and every optimal dual solution induces the unique primal optimizer, the resulting transition tensors converge linearly to
$\{\mathscr{B}^{(k)\star}\}_{k=2}^{K}$.
\end{proof}

In summary, the Sinkhorn--Schr\"odinger scaling provides an efficient and theoretically sound method to compute the optimal broadcasting transition tensors for non-uniform directed hypergraphs. By parametrizing the tensors with strictly positive scaling vectors and enforcing the pivot and receiver marginals via alternating multiplicative updates, the algorithm guarantees convergence to a unique solution. This approach not only preserves the structural constraints of the KL projection but also enables linear convergence analysis through the dual block coordinate ascent framework, making it practical for large-scale hypergraph inference.

\section{MERWs for Higher-Order Merging}\label{sec:merge}
We next consider the many-to-one merging setting, in which each directed hyperedge represents a multiway interaction that resolves a group of nodes into a single outcome, which can be seen as the reverse mechanism of the broadcasting case. This abstraction naturally models a wide range of processes, including consensus formation in social or biological systems, data or signal fusion in sensor networks, and decision or action selection in collective control and robotics \cite{olfati2007consensus,olfati2006flocking,sahasrabuddhe2021modelling}. In such scenarios, the hyperedge captures the joint effect of multiple participants acting simultaneously, while the directionality encodes the resolution of this interaction into a unique result. In contrast to broadcasting, where projection yields a linear Markov recursion on node marginals, merging produces an intrinsically higher-order update: the next-node distribution depends on the current distribution through a polynomial map. Our objective is therefore to preserve this hyperedge-level semantics explicitly and to infer a consistent merging transition law via a maximum-entropy principle.

\begin{figure}
    \centering
    \includegraphics[width=.6\linewidth]{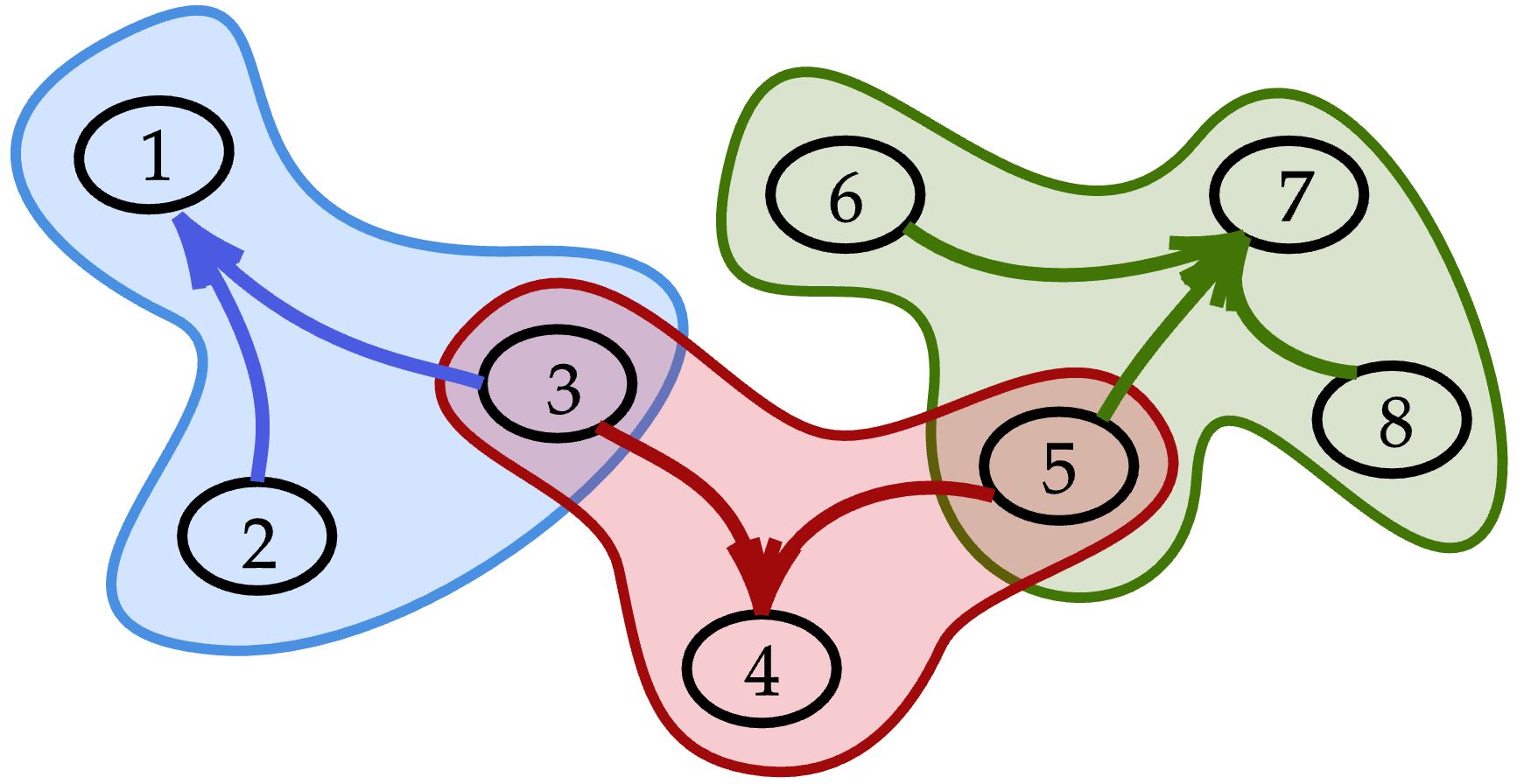}
    \caption{Merging MERWs on directed hypergraphs. Each colored region represents a directed hyperedge, with arrows indicating activation from a set of pivot (tail) nodes to a receiver (head) node (e.g., $\{v_2,v_3\}\rightarrow v_1$).}
    \label{fig:2}
\end{figure}

\subsection{Merging MERWs on Uniform Hypergraphs}

Consider a $k$-uniform directed hypergraph with $n$ nodes, in which each hyperedge has $k-1$ tail nodes and one head node. The merging transition tensor is a nonnegative, $k$th-order, $n$-dimensional back-symmetric tensor  $\mathscr{M}\in\mathbb R_+^{n\times n\times\stackrel{k}{\cdots} \times n}$ such that for the set of pivot (tail) nodes $\{v_{j_1},v_{j_2},\dots,v_{j_{k-1}}\}$, the entries $\mathscr{M}_{j_1 j_2 \cdots j_k}$ specify the probability that the interaction resolves to the receiver (head) node $v_{j_k}$.

\begin{definition}[Merging]
Let $\mathscr{M}\in\mathbb R_+^{n\times n\times\stackrel{k}{\cdots} \times n}$ be a merging transition tensor and  $\textbf{p}_t \in \mathbb{R}_+^n$ be the state distribution  at time $t$. The merging dynamics, which induce the node-level recursion, is defined as
\begin{equation}\label{eq:uniform-merge-dynamics}
\textbf{p}_{t+1} = \mathscr{M} \times_{1,2,\dots,k-1} \{\textbf{p}_t,\textbf{p}_t,\stackrel{k-1}{\dots} \textbf{p}_t\},
\end{equation}
where the tensor vector contractions aggregate over all $(k-1)$-tuples of pivot nodes weighted by their joint occurrence under $\textbf{p}_t$, and map each such multiway interaction to its corresponding receiver.
\end{definition}

In contrast to broadcasting dynamics, which reduce to linear recursions on node marginals, merging dynamics inherently preserves the higher-order structure through this homogeneous polynomial map. The merging MERW inference problem parallels the broadcasting case at the hyperedge level, but with higher-order interactions. We impose two key requirements. First, each set of pivot nodes must induce a valid probability distribution over the receiver node, enforced via the output-stochastic constraint defined below. Second, the induced node-level recursion must admit a prescribed stationary distribution, ensuring consistency between the hyperedge-level transition law and the desired equilibrium behavior. Specifically, let $\mathcal{Y}_+$ denote the cone of nonnegative, $k$th-order, $n$-dimensional back-symmetric tensors. Fix $\textbf{p} \in \mathbb{R}^n_{++}$ with $\textbf{1}^\top \textbf{p} = 1$. We then define the admissible class of merging kernels as
\begin{equation}\label{eq:merge-constraint}
\begin{split}
\mathcal M_k(\bp)
=
\Big\{\mathscr{M} \in \mathcal Y_+ \ \Big|\
\mathscr{M}\times_k \textbf{1}=\mathscr{I},\mathscr{M}\times_{1,2,\dots,k-1}\{\textbf{p},\textbf{p},\stackrel{k-1}{\dots},\textbf{p}\}=\textbf{p}
\Big\},
\end{split}
\end{equation}
where $\mathscr{I}\in\mathbb{R}^{n\times n\times\stackrel{k-1}{\cdots} \times n}$ denotes the all-one tensor.

We now formulate maximum-entropy inference for the merging kernel. Given a hypergraph support tensor $\mathscr{A}$ (e.g., the degree-normalized adjacency tensor) and a target stationary distribution $\textbf{p}$, we select the merging transition tensor $\mathscr{M}$ that minimizes the relative entropy with respect to $\mathscr{A}$ among all kernels satisfying the constraints in \eqref{eq:merge-constraint}.

\begin{problem} \label{prob:3}
Given a $k$-uniform directed hypergraph with a single head per hyperedge and $n$ nodes, fix a strictly positive distribution $\textbf{p} \in \mathbb{R}^n_{++}$ with $\textbf{1}^\top \textbf{p} = 1$, and let $\mathscr{A} \in \mathbb{R}_+^{n \times n \times \stackrel{k}{\cdots} \times n}$ denote the degree-normalized adjacency tensor. The corresponding merging transition tensor $\mathscr{M} \in \mathbb{R}_+^{n \times n \times \stackrel{k}{\cdots} \times n}$ is obtained by solving
\begin{align}\label{eq:obj3}
\min_{\mathscr{M}\in \mathcal M(\textbf{p})}\ 
\sum_{j_1=1}^n\sum_{j_2=1}^n\cdots\sum_{j_k=1}^n \mathscr{M}_{j_1j_2\cdots j_k}
\log\!\Big(\frac{\mathscr{M}_{j_1j_2\cdots j_k}}{\mathscr{A}_{j_1 j_2\cdots j_k}}\Big),
\end{align}
subject to the support constraint $\support(\mathscr{M}) \subseteq \support(\mathscr{A})$.
\end{problem}

Similar to the broadcasting setting, we solve \eqref{eq:obj3} directly over the merging transition tensor $\mathscr{M}$. To facilitate the KL projection, we introduce a pivot-weighted reference tensor supported on the hyperedges defined as
\begin{equation*}
\mathscr{K}_{j_1j_2\cdots j_k}=\prod_{p=1}^{k-1}\textbf{p}_{j_p}\mathscr{A}_{j_1j_2\cdots j_k},
\quad
\support(\mathscr{K})=\support(\mathscr{A}),
\end{equation*}
preserves the feasible set and does not alter the optimizer of \eqref{eq:obj3}, but allows the merging MERW to be expressed as a KL projection relative to $\mathscr{K}$.

\begin{proposition}\label{prop:4}
Assume that the constraint set $\mathcal{M}(\textbf{p})$ has nonempty relative interior with respect to the prescribed support.
Problem~3 admits a unique optimal merging transition tensor
$\mathscr{M}^{\star}\in\mathcal{M}(\textbf{p})$.
Moreover, there exist a strictly positive symmetric tensor
$\mathscr{U}\in
\mathbb{R}^{n\times n\times\stackrel{k-1}{\cdots}\times n}_{++}$
and a vector $\textbf{v}\in\mathbb{R}^{n}_{++}$ such that
$\mathscr{M}^{\star}$ admits the multiplicative factorization

\begin{equation}
\mathscr{M}^{\star}_{j_1j_2\cdots j_k}
=
\mathscr{K}_{j_1j_2\cdots j_k}
\mathscr{U}_{j_1j_2\cdots j_{k-1}}
\textbf{v}_{j_k}^{\omega_{j_1j_2\cdots j_{k-1}}},
\label{eq:merging_factorization}
\end{equation}
where $\omega_{j_1j_2\cdots j_{k-1}}=\prod_{p=1}^{k-1}\textbf{p}_{j_p}$ denotes the  pivot weight.
\black
\end{proposition}

\begin{proof}
Existence of $\mathscr{M}^{\star}$ follows from the compactness of the
feasible set and the continuity of the KL objective, while uniqueness
follows from the strict convexity of the KL objective over the
prescribed support, by the same argument as in Proposition~1.
Moreover, since $\mathcal{M}(\textbf{p})$ has nonempty relative
interior with respect to the prescribed support, the optimal tensor
$\mathscr{M}^{\star}$ is strictly positive on
$\operatorname{Supp}(\mathscr{K})$.

To derive the multiplicative form, we rewrite~(12) relative to the
pivot-weighted reference tensor $\mathscr{K}$ and introduce Lagrange
multipliers
$\mathscr{R}\in
\mathbb{R}^{n\times n\times\stackrel{k-1}{\cdots}\times n}$
and $\boldsymbol{\beta}\in\mathbb{R}^{n}$ corresponding to the two
constraints in~(11).
The Lagrangian is
\begin{align*}
\mathcal L(&\mathscr{M},\mathscr{R},\boldsymbol{\beta})
=
\sum_{j_1=1}^n\sum_{j_2=1}^n\cdots\sum_{j_k=1}^n \mathscr{M}_{j_1j_2\cdots j_k}
\log\Big(\frac{\mathscr{M}_{j_1j_2\cdots j_k}}{\mathscr{K}_{j_1j_2\cdots j_k}}\Big)\\
&+\sum_{j_1=1}^n\sum_{j_2=1}^n\cdots\sum_{j_{k-1}=1}^n
\mathscr{R}_{j_1j_2\cdots j_{k-1}}
\Big(\sum_{j_k=1}^n\mathscr{M}_{j_1j_2\cdots j_k}-1\Big)\\
&
+\sum_{j_k=1}^n\boldsymbol{\beta}_{j_k}
\Big(\sum_{j_1=1}^n\sum_{j_2=1}^n\cdots\sum_{j_{k-1}=1}^n\prod_{p=1}^{k-1}\textbf{p}_{j_p}\mathscr{M}_{j_1j_2\cdots j_k}-\textbf{p}_{j_k}\Big).
\end{align*}
At the optimum, for any index
$(j_1,j_2,\ldots,j_k)\in\operatorname{Supp}(\mathscr{K})$,
the first-order condition gives
\begin{equation*} 
\displaystyle
\frac{\partial \mathcal L}{\partial \mathscr{M}_{j_1j_2\cdots j_k}}
=
\log\Big(\frac{\mathscr{M}^\star_{j_1j_2\cdots j_k}}{\mathscr{K}_{j_1j_2\cdots j_k}}\Big)+1
+\mathscr{R}_{j_1j_2\cdots j_{k-1}}+\prod_{p=1}^{k-1}\textbf{p}_{{j_p}}\boldsymbol{\beta}_{j_k}
=0.
\end{equation*}

Using 
$\omega_{j_1j_2\cdots j_{k-1}}=\prod_{p=1}^{k-1}\textbf{p}_{j_p}$,
this condition becomes
\begin{equation*}
\log\left(
\frac{\mathscr{M}^{\star}_{j_1j_2\cdots j_k}}
{\mathscr{K}_{j_1j_2\cdots j_k}}
\right)
=
-1
-\mathscr{R}_{j_1j_2\cdots j_{k-1}}
-\omega_{j_1j_2\cdots j_{k-1}}
\boldsymbol{\beta}_{j_k}.
\end{equation*}
Exponentiating both sides yields
\begin{align*}
\mathscr{M}^{\star}_{j_1j_2\cdots j_k}
=
\mathscr{K}_{j_1j_2\cdots j_k}
\exp\left(
-1-\mathscr{R}_{j_1j_2\cdots j_{k-1}}
\right) \times 
\exp\left(
-\omega_{j_1j_2\cdots j_{k-1}}
\boldsymbol{\beta}_{j_k}
\right).
\end{align*}
Define the positive scaling tensor
$\mathscr{U}\in
\mathbb{R}^{n\times n\times\stackrel{k-1}{\cdots}\times n}_{++}$ with  
$\mathscr{U}_{j_1j_2\cdots j_{k-1}}
=
\exp(
-1-\mathscr{R}_{j_1j_2\cdots j_{k-1}}
)$
and vector $\textbf{v}\in\mathbb{R}^{n}_{++}$ with 
$\textbf{v}_{j}
=
\exp(
-\boldsymbol{\beta}_{j}
)$. Since
$
\exp(
-\omega_{j_1j_2\cdots j_{k-1}}
\boldsymbol{\beta}_{j_k}
)
=
\textbf{v}_{j_k}^{\omega_{j_1j_2\cdots j_{k-1}}},
$
we obtain
\begin{equation*}
\mathscr{M}^{\star}_{j_1j_2\cdots j_k}
=
\mathscr{K}_{j_1j_2\cdots j_k}
\mathscr{U}_{j_1j_2\cdots j_{k-1}}
\textbf{v}_{j_k}^{\omega_{j_1j_2\cdots j_{k-1}}}.
\end{equation*}
Finally, since $\mathscr{M}^{\star}$ and $\mathscr{K}$ are back-symmetric and
$\omega_{j_1j_2\cdots j_{k-1}}$ is invariant under permutations of the
indices, the scaling tensor $\mathscr{U}$ is symmetric, which completes the proof. 
\end{proof}

The result guarantees that the optimal maximum entropy merging transition tensor $\mathscr{M}^\star$ is unique and preserves the hypergraph topology encoded in $\mathscr{K}$, adjusting only the transition weights to enforce the output-stochasticity constraint and the prescribed stationary distribution induced by the homogeneous polynomial recursion. Unlike broadcasting dynamics, which reduce to a linear recursion on node marginals, merging dynamics inherently retain higher-order interactions at the node level. This distinction is reflected in the multiplicative, back-symmetric factorization of $\mathscr{M}^\star$, which decouples input group scaling from output scaling and mirrors the nonlinear structure of the polynomial map. At the same time, this structured multiplicative form naturally supports efficient iterative schemes based on alternating rescaling, making the computation of $\mathscr{M}^\star$ both principled and tractable.

\subsection{Extension to Non-Uniform Hypergraphs}
Extending merging MERWs to non-uniform hypergraphs follows a parallel approach to the broadcasting case, generalizing the transition tensor and constraints to accommodate hyperedges of varying sizes while preserving output-stochasticity and the prescribed stationary distribution. Specifically, for each $k \in \{2,3,\dots,K\}$, let $\mathscr{A}^{(k)} \in \mathbb{R}_+^{n \times n\times \stackrel{k}{\cdots} \times n}$ denote the degree-normalized adjacency tensors of $k$-uniform hyperedges, and let $\mathscr{M}^{(k)} \in \mathbb{R}_+^{n \times n\times \stackrel{k}{\cdots} \times n}$ be the corresponding merging transition tensors. To capture the relative influence of different hyperedge sizes on the overall dynamics, we introduce nonnegative layer weights $\{\lambda_k\}_{k=2}^{K}$ with summation equal to one, which modulate each layer’s contribution to the effective node-level recursion. Layers with zero weight are omitted from the joint optimization, and a pure-layer endpoint is treated as the corresponding single-layer problem.

Let $\textbf{p}_t\in\mathbb R_+^n$ denote the state distribution at time $t$. The induced node-level merging dynamics on a non-uniform hypergraph is given by the weighted superposition
\begin{align}\label{eq:nonuniform-merge-dynamics}
\textbf{p}_{t+1}
=
\sum^{K}_{ k=2}
\lambda_k\,
\mathscr{M}^{(k)}\times_{1,2,\dots,k-1}\{\textbf{p}_t,\textbf{p}_t,\stackrel{k-1}{\dots},\textbf{p}_t\}.
\end{align}
To define the admissible class of non-uniform merging kernels, let $\mathcal Y^{(k)}_+$ denote the cone of nonnegative $k$th-order back-symmetric tensors, and
fix $\textbf{p}\in\mathbb R^n_{++}$ with $\textbf{1}^\top\textbf{p}=1$. The admissible set is then defined as
\begin{equation}\label{eq:constraint2-merge}
\begin{split}
\mathcal M_{\le K}(\bp)
=
\Big\{\{
\mathscr{M}^{(k)}\}_{k=2}^K \ \mid \ \mathscr{M}^{(k)}\in\mathcal Y^{(k)}_+,\
\mathscr{M}^{(k)}\times_k\textbf{1}
=
\mathscr{I},
\sum_{k=2}^{K}
\lambda_k
\mathscr{M}^{(k)}\times_{1,2,\dots,k-1}\{\textbf{p},\textbf{p},\stackrel{k-1}{\dots},\textbf{p}\}
=
\textbf{p}
\Big\}.
\end{split}
\end{equation}
\begin{problem} \label{prob:nonuniform-merge}
Given a directed hypergraph with a single head per hyperedge, $n$ nodes, and maximum hyperedge size $K$, fix a strictly positive distribution $\textbf{p} \in \mathbb{R}^n_{++}$ with $\textbf{1}^\top \textbf{p} = 1$, and let $\mathscr{A}^{(k)} \in \mathbb{R}_+^{n \times n \times \stackrel{k}{\cdots} \times n}$ denote the degree-normalized adjacency tensors for $k$-uniform layers. The corresponding merging transition tensors $\mathscr{M}^{(k)} \in \mathbb{R}_+^{n \times n \times \stackrel{k}{\cdots} \times n}$ are obtained by solving
\begin{equation}\label{eq:obj2-merge}
\displaystyle\min_{{\{\mathscr{M}^{(k)}\}_{k=2}^{K}}\in \mathcal M_{\leq K}(\textbf{p})}\sum_{k={ 2}}^K\lambda_k\ 
\sum_{j_1=1}^n\sum_{j_2=1}^n\cdots\sum_{j_k=1}^n \mathscr{M}^{(k)}_{j_1j_2\cdots j_k}
\log\!\Big(\frac{\mathscr{M}^{(k)}_{j_1j_2\cdots j_k}}{\mathscr{A}^{(k)}_{j_1 j_2\cdots j_k}}\Big),
\end{equation}
subject to the support constraint $\support(\mathscr{M}^{(k)}) \subseteq \support(\mathscr{A}^{(k)})$ for $k=2,3,\dots,K$.
\end{problem}

\begin{proposition}\label{prop:5}
Assume that the constraint set $\mathcal M_{\le K}(\textbf{p})$ has nonempty relative interior with respect to the prescribed support. Problem~\ref{prob:nonuniform-merge} admits a unique set of optimal merging transition tensors $\{\mathscr{M}^{(k)\star}\}_{k=2}^K \in \mathcal M_{\le K}(\textbf{p})$. Moreover, there exist strictly positive symmetric tensors $\mathscr{U}^{(k)}\in\mathbb{R}^{n\times n\times\stackrel{k-1}{\cdots} \times n}_{++}$ and a strictly positive vector $\textbf{v}\in\mathbb{R}^n_{++}$ such that $\mathscr{M}^{(k)\star}$ admits the multiplicative factorization

\begin{equation}\label{eq:close-form2-merge}
\mathscr{M}^{(k)\star}_{j_1j_2\cdots j_k}
=
\mathscr{K}^{(k)}_{j_1j_2\cdots j_k}
\mathscr{U}^{(k)}_{j_1j_2\cdots j_{k-1}}
\textbf{v}_{j_k}^{\omega^{(k)}_{j_1j_2\cdots j_{k-1}}},
\end{equation}
\black
where 
$\omega^{(k)}_{j_1j_2\cdots j_{k-1}}
=
\prod_{p=1}^{k-1}\textbf{p}_{j_p}$ denotes the pivot weight,
\black
and $\mathscr{K}^{(k)}$ denote the pivot-weighted reference tensors for $k$-uniform layers, for $k=2,3,\dots,K$.
\end{proposition}

\begin{proof}
The existence and uniqueness arguments follow from those in
Propositions~\ref{prop:pro2} and~\ref{prop:4}. To derive the
multiplicative form, let $\mathscr{R}^{(k)}$ denote the Lagrange
multiplier associated with the output-stochasticity constraint of the
$k$-uniform layer and let $\boldsymbol{\beta}\in\mathbb{R}^{n}$ denote
the multiplier associated with the common stationarity constraint.

Using
$\omega^{(k)}_{j_1j_2\cdots j_{k-1}}
=
\prod_{p=1}^{k-1}\textbf{p}_{j_p}$,
the KKT condition for each $k$-uniform layer gives
\begin{equation*}
\lambda_k
\Big[
\log\big(
\frac{
\mathscr{M}^{(k)\star}_{j_1j_2\cdots j_k}
}{
\mathscr{K}^{(k)}_{j_1j_2\cdots j_k}
}
\big)+1
\Big]
+
\mathscr{R}^{(k)}_{j_1j_2\cdots j_{k-1}}
+
\lambda_k
\omega^{(k)}_{j_1j_2\cdots j_{k-1}}
\boldsymbol{\beta}_{j_k}
=0.
\end{equation*}
For $\lambda_k>0$, exponentiating the KKT condition yields
\eqref{eq:close-form2-merge}, where
$\mathscr{U}^{(k)}_{j_1 j_2 \cdots j_{k-1}}
=\exp(-1-\mathscr{R}^{(k)}_{j_1j_2\cdots j_{k-1}}/\lambda_k)$ and 
$\textbf{v}_{j}=\exp(-\boldsymbol{\beta}_{j})$. Since $\boldsymbol{\beta}$ is associated with the single common
stationarity constraint, $\textbf{v}$ is shared across all layers,
whereas $\mathscr{U}^{(k)}$ is layer-dependent. The symmetry of
$\mathscr{U}^{(k)}$ follows as in Proposition~\ref{prop:4}, which completes the proof.
\black
\end{proof}

Note that, unlike in the broadcasting case, the factor tensors  $\mathscr{U}^{(k)}$ are layer-dependent, reflecting the intrinsically polynomial nature of merging dynamics.  The receiver scaling vector $\textbf{v}$ is shared across layers due to the common stationarity constraint. Nonetheless, the formulation can still be interpreted as a one-step transport problem on a directed hypergraph.

\subsection{Ergodicity}

The merging dynamics induce nonlinear recursions on the simplex.  Let $\Delta_n:=\{\textbf{x}\in\mathbb{R}_+^n\mid \textbf{1}^{\top}\textbf{x}=1\}$ denote the probability simplex. Unlike broadcasting, the node marginals are not propagated by powers of a single Markov kernel. Instead, the evolution is governed by a polynomial map of degree $k-1$. Consequently, classical tools for analyzing ergodicity of linear Markov chains do not apply directly. Instead, the long-term behavior must be studied through the properties of nonlinear dynamical systems on the simplex. In particular, questions of ergodicity concern whether the polynomial map admits a unique stationary distribution and whether trajectories starting from arbitrary initial distributions converge to this fixed point. We first consider ergodicity for uniform merging. 

Based on the induced node-level uniform merging dynamics \eqref{eq:uniform-merge-dynamics}, define the homogeneous transition kernel
\begin{equation}\label{eq:homogeneous-kernel}
f(\textbf{p}) = \mathscr{M} \times_{1,2,\dots,k-1} \{\textbf{p},\textbf{p},\stackrel{k-1}{\dots} \textbf{p}\}.
\end{equation}
With this definition, the uniform merging dynamics can be written as $\textbf{p}_{t+1} = f(\textbf{p}_{t})$. {A stationary distribution therefore satisfies the tensor fixed-point equation $f(\textbf{p})=\textbf{p}$ which can be interpreted as a Z-eigenvector-type equation with eigenvalue one. In general, probability solutions to this equation need not be unique; consequently, neither uniqueness of the stationary distribution nor its global attraction is guaranteed without additional conditions \cite{benson2017spacey,fasino2020ergodicity,cui2026memory}.}

\begin{definition}[Strong Ergodicity \cite{saburov2019ergodicity}]\label{def:merge-strong-ergodic}
Let $\mathscr{M}$ be a merging transition tensor with recursion
$\textbf{p}_{t+1} = f(\textbf{p}_t)$.
The merging dynamics is said to be strongly ergodic if there exists a unique fixed point $\textbf{p}^\star\in \Delta_n$ such that $\textbf{p}_t\to\textbf{p}^\star$ for every initialization $\textbf{p}_0\in \Delta_n$.
\end{definition}

A convenient sufficient condition for strong ergodicity is contraction.
If $f$ is a strict contraction on $(\Delta_n,\|\cdot\|_1)$, then
Banach's fixed-point theorem guarantees the existence of a unique fixed point
and global convergence.
We work with the $\|\cdot\|_1$ norm, which induces the natural metric on $\Delta_n$
and is consistent with total variation arguments for polynomial stochastic operators
\cite{saburov2019ergodicity}. Since $\mathscr{M}$ is output-stochastic, each slice $\mathscr{M}_{j_1j_2\cdots j_{k-1}:}$ is a probability vector on the receiver node. Define the row-variation coefficient
\begin{equation}\label{eq:merge-rowvar}
\delta(\mathscr{M})
=
\max_{\substack{a_1,a_2,\dots,a_{k-1} \\b_1,b_2,\dots,b_{k-1}}}
\frac12\big\|\mathscr{M}_{a_1a_2\cdots a_{k-1}:}-\mathscr{M}_{b_1b_2\cdots b_{k-1}:}\big\|_1.
\end{equation}
This quantity represents the worst-case total variation distance between two conditional output distributions and coincides with Dobrushin's third-order ergodicity coefficient for stochastic tensors \cite[Definition 5.1(iii)]{saburov2019ergodicity}.

\begin{lemma}\label{lem:merge-Lip}
Let $\mathscr{M}\in \mathbb{R}_+^{n \times n \times \stackrel{k}{\cdots} \times n}$ be an output-stochastic merging transition tensor and let $f$ be the homogeneous transition kernel defined in \eqref{eq:homogeneous-kernel}.
Then
\begin{equation}\label{eq:merge-Lip}
\|f(\textbf{p})-f(\textbf{q})\|_1
\le
(k-1)\,\delta(\mathscr{M})\,\|\textbf{p}-\textbf{q}\|_1
\end{equation}
for all $\textbf{p},\textbf{q}\in\Delta_n$.
\end{lemma}

\begin{proof}
Since $f$ is multilinear in the first $k-1$ arguments, it can be shown that
\begin{equation*}
\displaystyle
f(\textbf{p})-f(\textbf{q})
=
\sum_{s=1}^{{ k-1}}
\mathscr M\times_{1,2,\dots,k-1}
\{\textbf{p},\dots,\textbf{p},\underbrace{(\textbf{p}-\textbf{q})}_{\text{$s$th}},\textbf{q},\dots,\textbf{q}\}.
\end{equation*}
Define the induced matrix $\textbf{P}^{(s)}\in\mathbb{R}^{n\times n}_{+}$ such that
\begin{equation*}
\displaystyle
    \textbf P^{(s)}
=
\mathscr M\times_{1,2,\dots,s-1}\{\textbf{p},\textbf{p}\dots,\textbf{p}\}\times_{s+1,s+2,\dots,k-1}\{\textbf{q},\textbf{q},\dots,\textbf{q}\}.
\end{equation*}
Then the $s$th term in the above summation can be written as $\textbf P^{(s)\top}(\textbf{p}-\textbf{q})$. By the output-stochasticity of $\mathscr M$, each $\textbf P^{(s)}$ is row-stochastic. Moreover, each row of $\textbf P^{(s)}$ is a convex combination of the conditional laws $\mathscr M_{a_1a_2\cdots a_{k-1}:}$, so the Dobrushin coefficient of $\textbf P^{(s)}$ is bounded by $\delta(\mathscr M)$. Therefore, the standard Dobrushin contraction bound for row-stochastic matrices gives
\[
\|\textbf P^{(s)\top}(\textbf{p}-\textbf{q})\|_1
\le
\delta(\mathscr M)\,\|\textbf{p}-\textbf{q}\|_1.
\]
Applying this to each term in $f(\textbf{p})-f(\textbf{q})$ and summing over $s=1,2,\dots,k-1$ yields
\begin{align*}
\|f(\textbf{p})-f(\textbf{q})\|_1
&\le
(k-1)\,\delta(\mathscr M)\,\|\textbf{p}-\textbf{q}\|_1, 
\end{align*}
which proves the claim.
\end{proof}

This result shows that the homogeneous transition kernel $f$ is Lipschitz on $(\Delta_n, \|\cdot\|_1)$ with Lipschitz constant $(k-1)\delta(\mathscr{M})$. In particular, if this constant is strictly less than $1$, then $f$ is a strict contraction \cite[Theorem~5.4(ii)]{saburov2019ergodicity}. Moreover, the above lemma extends naturally to non-uniform merging dynamics. Similarly, define the polynomial transition kernel 
\begin{equation*}
    \displaystyle  F(\textbf{p})=\sum^K_{k=2}\lambda_k\mathscr{M}^{(k)} \times_{1,2,\dots,k-1} \{\textbf{p},\textbf{p},\stackrel{k-1}{\dots} \textbf{p}\} = \sum^K_{ k=2}\lambda_k f_k(\textbf{p}).
\end{equation*}
\begin{lemma}
    If $\sum^K_{ k=2}\lambda_k\,(k-1)\,\delta(\mathscr{M}^{(k)}) < 1$,
then $F$ is a strict contraction on $(\Delta_n,\|\cdot\|_1)$. Consequently, $F$ admits a unique fixed point $\textbf{p}^\star\in\Delta_n$, and for every initialization $\textbf{p}_0\in\Delta_n$, it follows that
\[
\|\bp_t-\bp^\star\|_1
\le
\Big(\sum_{k=2}^{K}\lambda_k\,(k-1)\,\delta(\mathscr{M}^{(k)})\Big)^t\,\|\bp_0-\bp^\star\|_1.
\]
In particular, the non-uniform merging recursion is strongly ergodic in the sense of Definition~\ref{def:merge-strong-ergodic}.
\end{lemma}
\begin{proof}
    By convexity of $\|\cdot\|_1$ and the definition of $F$, it can be shown that
\begin{align*}
\|F(\textbf{p})-F(\textbf{q})\|_1
&=
\Big\|\sum^K_{ k=2}\lambda_k\big(f_k(\textbf{p})-f_k(\textbf{q})\big)\Big\|_1\\
&\le
\sum^K_{ k=2}\lambda_k\,\|f_k(\textbf{p})-f_k(\textbf{q})\|_1.    
\end{align*}
Applying Lemma~\ref{lem:merge-Lip} to each layer yields
\[
\|f_k(\textbf{p})-f_k(\textbf{q})\|_1
\le
(k-1)\,\delta(\mathscr{M}^{(k)})\,\|\textbf{p}-\textbf{q}\|_1,
\]
and hence
\[
\|F(\textbf{p})-F(\textbf{q})\|_1
\le
\Big(\sum^K_{ k=2}\lambda_k\,(k-1)\,\delta(\mathscr{M}^{(k)})\Big)\,\|\textbf{p}-\textbf{q}\|_1.
\]
Therefore, if $\sum_{k=2}^{K}\lambda_k\,(k-1)\,\delta(\mathscr{M}^{(k)}) < 1$, $F$ is a strict contraction. The remaining claims follow from Banach's fixed-point theorem.
\end{proof}

The preceding lemmas establish that both uniform and non-uniform merging dynamics are strongly ergodic whenever the corresponding Lipschitz constants are strictly less than one. Under these conditions, the dynamics admits a unique fixed point $\textbf{p}^\star\in\Delta_n$, and all trajectories initialized in the probability simplex converge globally to this fixed point. Importantly, the polynomial, multilinear structure of merging dynamics ensures that long-term behavior is governed by higher-order aggregation, in contrast to broadcasting dynamics, where convergence is dictated solely by projected pairwise interactions and linear matrix theory. {Note that the contraction condition in Lemma~3 provides one sufficient criterion that rules out this ambiguity: when it holds, the prescribed stationary distribution is the unique fixed point in the simplex and is globally attracting. When the contraction condition is not satisfied, this conclusion no longer follows; the dynamics may admit multiple stationary distributions or exhibit dependence on the initial distribution. Failure of the condition does not, however, imply nonconvergence, since the condition is sufficient rather than necessary.}

\subsection{Sinkhorn--Schr\"odinger Scaling}

Similar to broadcasting, we can formulate  the non-uniform merging
 inference problem on general directed hypergraphs as a KL projection onto
two families of linear constraints.
Unlike in the broadcasting case, however, the stationarity constraint in
\eqref{eq:constraint2-merge} is a weighted marginal constraint, with the
weights determined by the given distribution $\textbf{p}$. Consequently,
as shown in Proposition~\ref{prop:5}, the receiver scaling enters through
the pivot-dependent exponent
$\omega^{(k)}_{j_1 j_2\cdots j_{k-1}}$ rather than an ordinary mode-wise
scaling. Since $\textbf{p}$ is fixed, these weights are fixed and the
constraints remain affine in the transition tensors.

For strictly positive layer-dependent scaling symmetric tensors
$\mathscr{U}^{(k)}
\in
\mathbb{R}^{n\times n\times\stackrel{k-1}{\cdots}\times n}_{++}$
and  a common vector $\textbf{v}\in\mathbb{R}^n_{++}$,
define the scaled merging transition tensors
according to \eqref{eq:close-form2-merge}, and denote them by
$\mathscr{M}^{(k)}(\mathscr{U}^{(k)},\textbf{v})$.
The scaling symmetric tensors $\mathscr{U}^{(k)}$ act on the pivot
modes, while the  common vector $\textbf{v}$ acts on the
receiver mode  through the power-weighted scaling. 
Under this parametrization, define
\begin{equation}\label{eq:marginals-merge}
\begin{aligned}
\zeta^{(k)}(\mathscr{U}^{(k)},\textbf{v})
&=
\mathscr{M}^{(k)}(\mathscr{U}^{(k)},\textbf{v})
\times_k\textbf{1},
\\
\xi(\{\mathscr{U}^{(k)}\}_{k=2}^{K},\textbf{v})
&=
\sum^{K}_{k= 2}
\lambda_k
\mathscr{M}^{(k)}(\mathscr{U}^{(k)},\textbf{v})
\times_{1,2,\dots,k-1}
\{\textbf{p},\textbf{p},\stackrel{k-1}{\dots},\textbf{p}\}.
\end{aligned}
\end{equation}
The two feasibility conditions in \eqref{eq:constraint2-merge} are
therefore equivalent to
\[
\zeta^{(k)}(\mathscr{U}^{(k)},\textbf{v})=\mathscr{I}
\quad\text{and}\quad
\xi(\{\mathscr{U}^{(k)}\}_{k=2}^{K},\textbf{v})
=\textbf{p}.
\]

For fixed $\textbf{v}$, the output marginal is affine in
$\mathscr{U}^{(k)}$. The corresponding KL projection gives the
usual Sinkhorn-type normalization
\begin{equation}\label{eq:merge-U-update-nonuniform}
\mathscr U^{(k)}
\leftarrow
\mathscr U^{(k)}
\odot
\left(
\mathscr I
\oslash
\zeta^{(k)}(\mathscr U^{(k)},\textbf v)
\right).
\end{equation}
This is the standard KL/Bregman projection onto the output-marginal
constraint.
For fixed $\{\mathscr{U}^{(k)}\}_{k=2}^{K}$, the receiver marginal is
not affine in $\textbf{v}$ because of the pivot-dependent exponent. Nevertheless, its $j$th component depends
only on the scalar $\textbf{v}_j$. For $s>0$, define
\begin{equation}\label{eq:Phi-merge}
\displaystyle
\Phi_j(s)
=
\sum_{k=2}^{K}\lambda_k
\sum_{j_1=1}^n\sum_{j_2=1}^n\cdots\sum_{j_{k-1}=1}^n
\omega^{(k)}_{j_1j_2\cdots j_{k-1}}
\mathscr{K}^{(k)}_{j_1j_2\cdots j_{k-1}j}
\mathscr{U}^{(k)}_{j_1j_2\cdots j_{k-1}}
s^{\omega^{(k)}_{j_1j_2\cdots j_{k-1}}}.
\end{equation}
By construction,
$\xi_j(\{\mathscr{U}^{(k)}\}_{k=2}^{K},\textbf{v})
=
\Phi_j(\textbf{v}_j)$.
Hence, the stationarity constraint reduces to $\Phi_j(\textbf{v}_j)=\textbf{p}_j$ independently for each receiver.

\begin{lemma}[Well-posedness of receiver scaling]\label{lem:Phi-merge}
Assume that $\mathcal{M}_{\le K}(\textbf{p})$ is nonempty. For fixed
strictly positive scaling tensors
$\{\mathscr{U}^{(k)}\}_{k=2}^{K}$, the function
$\Phi_j:\mathbb{R}_{++}\rightarrow\mathbb{R}_{++}$ defined in
\eqref{eq:Phi-merge} is a smooth, strictly increasing bijection.
Consequently, for  $j=1,2,\dots,n$, there exists a unique positive
receiver scaling
$
\textbf{v}_j=\Phi_j^{-1}(\textbf{p}_j).
$
\end{lemma}

\begin{proof}
Since $\textbf{p}\in\mathbb{R}^n_{++}$, all exponents
$\omega^{(k)}_{j_1j_2\cdots j_{k-1}}$ are strictly positive, and hence
$\Phi_j$ is smooth on $(0,\infty)$. Moreover, nonemptiness of
$\mathcal{M}_{\le K}(\textbf{p})$ and $\textbf{p}_j>0$ imply that, for
each receiver $j$, at least one coefficient in \eqref{eq:Phi-merge} is
strictly positive. Otherwise, the $j$th component of the stationarity
constraint could not be positive.
Differentiating \eqref{eq:Phi-merge} gives
\begin{equation*}
\begin{split}
\Phi_j'(s)
=
\sum_{k=2}^{K}\lambda_k
\sum_{j_1=1}^n\sum_{j_2=1}^n\cdots\sum_{j_{k-1}=1}^n
\left(
\omega^{(k)}_{j_1j_2\cdots j_{k-1}}
\right)^2\times
\mathscr{K}^{(k)}_{j_1j_2\cdots j_{k-1}j}
\mathscr{U}^{(k)}_{j_1j_2\cdots j_{k-1}}
s^{\omega^{(k)}_{j_1j_2\cdots j_{k-1}}-1}
>0 .
\end{split}
\end{equation*}
Therefore, $\Phi_j$ is strictly increasing on $(0,\infty)$. Since all exponents are strictly positive,  $\lim_{s\rightarrow0^+}\Phi_j(s)=0$. Since at least one coefficient is strictly positive, $\lim_{s\rightarrow+\infty}\Phi_j(s)=+\infty$.
Thus, $\Phi_j$ maps $\mathbb{R}_{++}$ continuously and strictly
monotonically onto $\mathbb{R}_{++}$. The intermediate value theorem gives
existence of a positive solution to
$\Phi_j(\textbf{v}_j)=\textbf{p}_j$, while strict monotonicity gives
uniqueness.
\end{proof}

Thus, receiver scaling update is given component-wise by
\begin{equation}\label{eq:merge-v-update-nonuniform}
\textbf{v}_j
\leftarrow
\Phi_j^{-1}(\textbf{p}_j).
\end{equation}
In general, $\Phi_j^{-1}$ need not admit a closed-form expression.
By Lemma~\ref{lem:Phi-merge}, $\Phi_j$ is continuous and strictly
increasing from $0$ to $+\infty$, so $\textbf{v}_j$ can be computed
uniquely by solving
$\Phi_j(s)-\textbf{p}_j=0$.
A valid bracket is obtained by geometric expansion from $s=1$ by
successively doubling or halving the endpoint as needed. Finite
termination follows from the limiting properties of $\Phi_j$. We use
Brent's method \cite{brent1971algorithm} for the resulting scalar
solve. For fixed $\{\mathscr{U}^{(k)}\}_{k=2}^{K}$, the receiver
equations are decoupled and can be solved independently.

These updates iteratively enforce the two constraints by normalizing
the pivot and receiver marginals in turn. Upon convergence, the
resulting tensors
\[
\mathscr{M}^{(k)\star}
=
\mathscr{M}^{(k)}
(\mathscr{U}^{(k)\star},\textbf{v}^{\star})
\]
for all $k=2,3,\dots,K$ constitute the solution to the KL projection
problem. Algorithm~\ref{alg:sinkhorn-merge-nonuniform} summarizes this
procedure. In particular, the
$\mathscr{U}^{(k)}$-updates impose output-stochasticity in each layer,
while the $\textbf v$-updates impose the stationarity
constraint through  the weighted receiver marginal.
Unlike broadcasting, the scaling tensors
$\mathscr{U}^{(k)}$ are layer-dependent, while the receiver
scaling vector $\textbf v$ is shared across all layers.

\begin{proposition}[Convergence]
\label{prop:bregman-merge-convergence}
Under the regularity conditions for dual coordinate ascent,
Algorithm~\ref{alg:sinkhorn-merge-nonuniform} admits  a fixed point
$(\{\mathscr{U}^{(k)\star}\}_{k=2}^{K},\textbf{v}^{\star})$
corresponding to the unique optimal merging transition tensors,
and the  induced transition tensors converge to the  optimal
solution at a linear rate.
\end{proposition}

\begin{proof}
The proof follows the same dual viewpoint as in the broadcasting case.
The updates \eqref{eq:merge-U-update-nonuniform} and \eqref{eq:merge-v-update-nonuniform} alternately enforce the two affine constraint families in \eqref{eq:constraint2-merge}.
To characterize the convergence and its rate, consider the dual objective.
Up to an additive constant, it can be written in terms of the positive
scaling variables as
\begin{equation*}
\begin{aligned}
\mathcal{G}(\{\mathscr{U}^{(k)}\},\textbf{v})
={}&
-\sum_{k=2}^{K}\lambda_k
\sum_{j_1=1}^n\sum_{j_2=1}^n\cdots\sum_{j_k=1}^n
\mathscr{K}^{(k)}_{j_1j_2\cdots j_k}
\mathscr{U}^{(k)}_{j_1j_2\cdots j_{k-1}}
\textbf{v}_{j_k}^{\omega^{(k)}_{j_1j_2\cdots j_{k-1}}}
\\
&+
\sum_{k=2}^{K}\lambda_k
\sum_{j_1=1}^n\sum_{j_2=1}^n\cdots\sum_{j_{k-1}=1}^n
\log
\mathscr{U}^{(k)}_{j_1j_2\cdots j_{k-1}}
+
\sum_{j=1}^{n}
\textbf{p}_j\log\textbf{v}_j .
\end{aligned}
\end{equation*}
Viewed in the logarithmic scaling variables, differentiation with
respect to the pivot scaling gives
\begin{equation*}
\frac{\partial\mathcal{G}}
{\partial\log\mathscr{U}^{(k)}_{j_1j_2\cdots j_{k-1}}}
=
\lambda_k
\left(
1-\zeta^{(k)}_{j_1j_2\cdots j_{k-1}}
\right).
\end{equation*}
Therefore,
\eqref{eq:merge-U-update-nonuniform} satisfies the corresponding
block optimality condition exactly.

Similarly, differentiation with respect to the logarithmic receiver
scaling gives
$
\frac{\partial\mathcal{G}}
{\partial\log\textbf{v}_j}
=
\textbf{p}_j-\Phi_j(\textbf{v}_j).
$
The receiver block is strictly concave since
$
\frac{\partial^2\mathcal{G}}
{\partial(\log\textbf{v}_j)^2}
=
-\textbf{v}_j\Phi_j'(\textbf{v}_j)
<0.
$
Hence, Lemma~\ref{lem:Phi-merge} shows that
$\textbf{v}_j=\Phi_j^{-1}(\textbf{p}_j)$ is the unique maximizer of
each receiver block. Therefore, each step of
Algorithm~\ref{alg:sinkhorn-merge-nonuniform} is an exact block
maximization of the dual, and the iteration is an exact block
coordinate ascent method.
After entries outside the prescribed support are fixed at zero and
omitted, the support-relative-interior condition in
Proposition~\ref{prop:5}, together with the positivity of all included
layer weights, ensures that all block updates are well defined.  Under the regularity conditions of Theorem 2.1 in \cite{luo1992convergence}, the dual iterates converge at least linearly to the optimal dual solution set, measured per complete scaling cycle. Since the primal optimizer is unique, the corresponding induced transition tensors converge at least linearly to the unique primal optimizer.
\end{proof}

\begin{algorithm}[t]
\caption{Sinkhorn--Schr\"odinger Scaling for Merging}
\label{alg:sinkhorn-merge-nonuniform}
\begin{algorithmic}[1]
\Require State distribution $\textbf{p}\in\mathbb R^n_{++}$ with
$\textbf{1}^{\top}\textbf{p}=1$, adjacency tensors
$\{\mathscr A^{(k)}\}_{k=2}^{K}$, weights
$\{\lambda_k\}_{k=2}^{K}$, tolerance $\varepsilon>0$

\State initialize $\mathscr U^{(k)}\gets\mathscr I$
(the all-one tensor)  for $k=2,3,\dots,K$
and $\textbf v\gets\textbf 1$

\While{not converged}

\For{$k=2,3,\dots,K$}

\State Compute
$\zeta^{(k)}(\mathscr U^{(k)},\textbf v)$
based on \eqref{eq:marginals-merge}

\State
$\mathscr U^{(k)}
\leftarrow
\mathscr U^{(k)}
\odot
\left(
\mathscr I
\oslash
\zeta^{(k)}(\mathscr U^{(k)},\textbf v)
\right)$

\EndFor

\For{$j=1,2,\dots,n$}

\State
$\textbf v_j
\leftarrow
\Phi_j^{-1}(\textbf p_j)$

\EndFor

\EndWhile

\State \textbf{return}
$\mathscr M^{(k)\star}
\gets
\mathscr M^{(k)}(\mathscr U^{(k)},\textbf v)$
for $k=2,3,\dots,K$

\end{algorithmic}
\end{algorithm}

In summary, the Sinkhorn--Schr\"odinger scaling provides an efficient
and principled way to compute the optimal merging transition tensors for
non-uniform directed hypergraphs. By parameterizing each layer with
strictly positive scaling tensors on the pivot nodes and
 a common scaling vector on the receiver node,
and enforcing the pivot and receiver marginals through alternating
multiplicative updates, the algorithm converges to the unique solution
of the KL projection problem. This procedure preserves the structural
constraints of the merging formulation, including layerwise
output-stochasticity and stationarity of the mixture dynamics. Moreover,
the resulting iterative scheme admits a natural interpretation within a
dual block coordinate ascent framework,  with the linear rate holding
under the stated regularity conditions,
making the method practical for large-scale hypergraph inference
involving higher-order merging interactions.

\black

\section{Numerical experiments}\label{sec:num}
All numerical experiments, including broadcasting MERWs, merging MERWs, and a real-world application on MovieLens next-item prediction, were conducted using MATLAB R2025b on a machine equipped with an M1 Pro CPU and 16 GB of memory. The code used for these experiments is available at { \url{https://github.com/dytroshut/hyperMERW}}.

\subsection{Broadcasting MERWs}

We first consider broadcasting on an $8$-node directed $3$-uniform
hypergraph, where each hyperedge has one tail node and two head nodes.
{ The $3$-uniform adjacency support is formed from the node triples of two groups
$
\{v_1,v_2,v_3\},\allowbreak
\{v_1,v_2,v_4\},\allowbreak
\{v_1,v_3,v_4\},\allowbreak
\{v_2,v_3,v_4\},
$
and
$
\{v_5,v_6,v_7\},\allowbreak
\{v_5,v_6,v_8\},\allowbreak
\{v_5,v_7,v_8\},\allowbreak
\{v_6,v_7,v_8\}.
$
For each listed triple, every node acts in turn as the tail node and
broadcasts to the other two nodes. For example, the triple $\{v_1,v_2,v_3\}$ generates three directed hyperedges $v_1\rightarrow \{v_2,v_3\}$, $v_2\rightarrow \{v_1,v_3\}$, and $v_3\rightarrow \{v_1,v_2\}$. We additionally include the two
cross-group directed hyperedges
$v_4\rightarrow\{v_6,v_7\}$ and
$v_6\rightarrow\{v_2,v_3\}$.}
We prescribe the stationary distribution
\[
\textbf{p}=
\begin{bmatrix}
0.07 & 0.07 & 0.07 & 0.07 & 0.18 & 0.18 & 0.18 & 0.18
\end{bmatrix}^{\top}
\]
such that $\textbf{1}^{\top}\textbf{p}=1$, and use the
degree-normalized adjacency tensor of the hypergraph as the
nonnegative reference tensor. We then compute the maximum-entropy
broadcasting transition tensor by solving Problem~\ref{prob:1} and
form the projected transition kernel, given by
\begin{equation*}
\textbf{P}=
\begin{bmatrix}
0     & 0.326 & 0.326 & 0.347 & 0     & 0     & 0     & 0\\
0.370 & 0     & 0.303 & 0.326 & 0     & 0     & 0     & 0\\
0.370 & 0.303 & 0     & 0.326 & 0     & 0     & 0     & 0\\
0.259 & 0.215 & 0.215 & 0     & 0     & 0.156 & 0.156 & 0\\
0     & 0     & 0     & 0     & 0     & 0.316 & 0.330 & 0.353\\
0     & 0.061 & 0.061 & 0     & 0.300 & 0     & 0.278 & 0.300\\
0     & 0     & 0     & 0     & 0.346 & 0.307 & 0     & 0.346\\
0     & 0     & 0     & 0     & 0.353 & 0.316 & 0.330 & 0
\end{bmatrix}.
\end{equation*}

\begin{figure}[t]
    \centering
    \includegraphics[width=.5\linewidth]{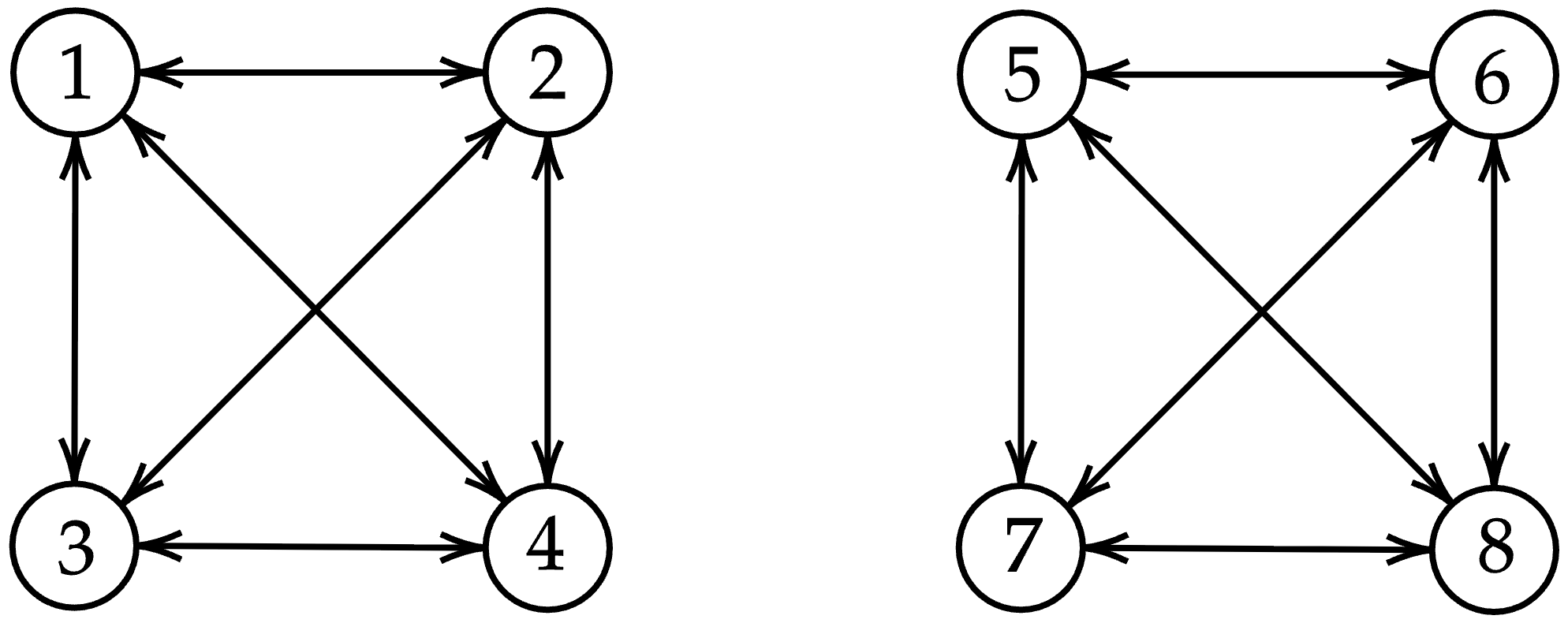}
    \caption{{ Pairwise adjacency support added in the non-uniform
    broadcasting example. Every directed edge among
    $v_1,v_2,v_3,v_4$ and among $v_5,v_6,v_7,v_8$, including a
    self-loop at each node, is included. The self-loops are omitted
    from the visualization for clarity.}}
    \label{fig:broadcast-pairwise-topology}
\end{figure}

\begin{figure}[t]
    \centering
    \includegraphics[width=.65\linewidth]{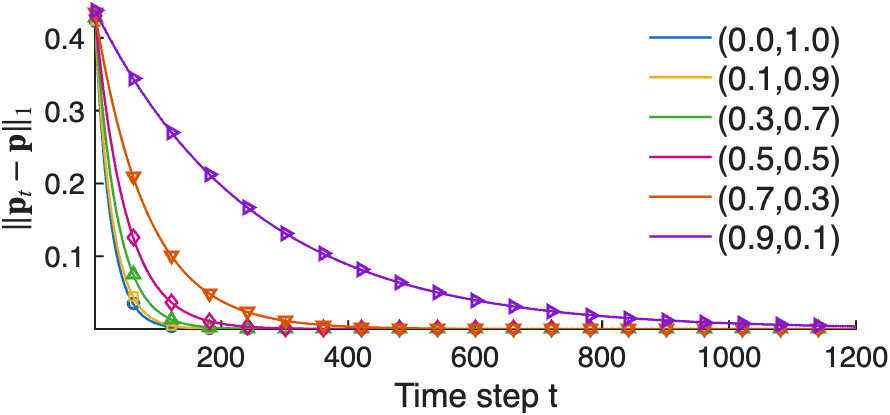}
    \caption{Mixing curves of the projected node-level dynamics for
    non-uniform broadcasting MERWs under different weights
    $\lambda_2$ and $\lambda_3$.}
    \label{fig:broadcast-mixing-weights}
\end{figure}

We next consider a non-uniform broadcasting model with two layers:
a pairwise layer ($k=2$) and a $3$-uniform layer ($k=3$).
{ We retain the $3$-uniform layer described above and add the
pairwise layer shown in Fig.~\ref{fig:broadcast-pairwise-topology}.}
We keep the same target distribution $\textbf{p}$. We compute the
maximum-entropy broadcasting transition tensor for each layer
jointly for each weight pair and obtain the projected
transition kernel by combining the induced node-level dynamics through
the convex mixture
\[
\textbf{P}(\lambda_2,\lambda_3)
=
\lambda_2\,\textbf{P}^{(2)}
+
\lambda_3\,\textbf{P}^{(3)},
\qquad
\lambda_2+\lambda_3=1,
\]
 where $\textbf{P}^{(2)}$ and $\textbf{P}^{(3)}$ denote the
projected layer contributions for the corresponding weight pair.
Fig.~\ref{fig:broadcast-mixing-weights} compares the mixing curves for
different choices of $\lambda_2$ and $\lambda_3$, demonstrating that
$\textbf{p}_t$ converges to the stationary distribution $\textbf{p}$.
In this example, larger values of $\lambda_3$ lead to faster
relaxation because the $3$-uniform layer provides the two cross-group
hyperedges that are absent from the pairwise layer.
Beyond encoding higher-order structure, the proposed framework reveals
how different interaction modes shape the transient dynamics while
remaining consistent with the same stationary distribution.

\subsection{Merging MERWs}

{ In this example, we consider merging on an $8$-node directed
 non-uniform hypergraph with a pairwise layer ($k=2$) and a
$3$-uniform layer ($k=3$), where each layer-$k$ hyperedge has
$k-1$ tail nodes and one head node. We use the same underlying non-uniform hypergraph supports as in the broadcasting example, with the hyperedge orientation changed to the merging form. For example, the triple $\{v_1,v_2,v_3\}$ generates three directed hyperedges $\{v_2,v_3\}\rightarrow v_1$, $\{v_1,v_3\}\rightarrow v_2$, and $\{v_1,v_2\}\rightarrow v_3$.} We prescribe the stationary
distribution
\begin{equation*}
\textbf{p}=
\begin{bmatrix}
0.172 & 0.226 & 0.003 & 0.146 & 0.102 & 0.081 & 0.115 & 0.156
\end{bmatrix}^{\top}
\end{equation*}
such that $\textbf{1}^{\top}\textbf{p}=1$.
{ The pairwise and $3$-uniform merging layers are specified by
the adjacency matrix $\mathscr{A}^{(2)}\in\mathbb{R}^{8\times 8}$ and the back-symmetric
adjacency tensor $\mathscr{A}^{(3)}\in\mathbb{R}^{8\times 8\times 8}$, respectively, whose nonzero
entries define the corresponding layer supports. 
The complete adjacency arrays are provided with the accompanying code and data. We compute the maximum-entropy merging transition
tensors jointly by solving
Problem~{\ref{prob:nonuniform-merge}}
{ for $(\lambda_2,\lambda_3)=(0.5,0.5)$}.}
{ The resulting pairwise transition matrix is 
\begin{equation*}
\begin{split}
\mathscr{M}^{(2)\star}
&=
\begin{bmatrix}
0.575 & 0.306 & 0     & 0     & 0     & 0.015 & 0     & 0.104\\
0.203 & 0.771 & 0     & 0.026 & 0     & 0     & 0     & 0\\
0     & 0.176 & 0.485 & 0     & 0.171 & 0     & 0.169 & 0\\
0     & 0.249 & 0     & 0.509 & 0.035 & 0.039 & 0     & 0.167\\
0     & 0     & 0     & 0.244 & 0.756 & 0     & 0     & 0\\
0.294 & 0     & 0     & 0.195 & 0     & 0.511 & 0     & 0\\
0     & 0     & 0     & 0     & 0     & 0     & 1.000 & 0\\
0.223 & 0     & 0     & 0.075 & 0     & 0     & 0     & 0.702
\end{bmatrix}.
\end{split}
\end{equation*}
Due to space limitations, the corresponding 3-uniform transition tensor $\mathscr{M}^{(3)\star}\in\mathbb{R}^{8\times 8\times 8}$ is omitted here and provided with the accompanying code and data.}

\black
We next vary the weights in this non-uniform merging model
 while keeping the two layer supports and the same target
distribution $\textbf{p}$. We compute the merging transition tensors
for both layers ($\mathscr{M}^{(2)}$ and $\mathscr{M}^{(3)}$)
 jointly for each strictly positive weight pair and combine
 their induced node-level maps as
\begin{equation*}
\textbf{p}_{t+1}
=
\lambda_2
\big(\mathscr{M}^{(2)}\big)^{\top}\textbf{p}_t
+
\lambda_3
\mathscr{M}^{(3)}
\times_{1,2}\{\textbf{p}_t,\textbf{p}_t\},
\qquad
\lambda_2+\lambda_3=1.
\end{equation*}
The choice $(\lambda_2,\lambda_3)=(0,1)$ is treated separately as the
pure $k=3$ case. Fig.~\ref{fig:merge-mixing-weights} compares the mixing curves for
different choices of $\lambda_2$ and $\lambda_3$, showing behavior
similar to the broadcasting example. {The trajectories
$\|\textbf{p}_t-\textbf{p}\|_1$, all initialized from the common
distribution $\textbf{p}_0=\frac{1}{8}\textbf{1}$, converge as the
time step increases.  For the tested weight choices, the sufficient condition in Lemma 3 is not satisfied, since the computed bounds range from $1.069305$ to $1.689073$. This condition is sufficient but not necessary for convergence, and the trajectories nevertheless converge numerically to the prescribed stationary distribution from the stated initialization.} Compared with broadcasting, the dependence on
$(\lambda_2,\lambda_3)$ is weaker in the merging case, as multiple
sources combine to produce a genuinely nonlinear macroscopic law.

\begin{figure}[t]
\centering
\includegraphics[width=.65\linewidth]{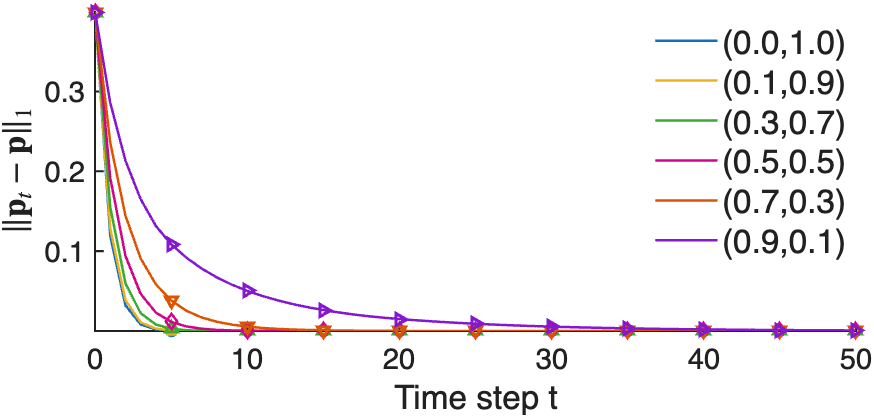}
\caption{Mixing curves of the node-level dynamics for non-uniform
merging MERWs under different weights $\lambda_2$ and $\lambda_3$.}
\label{fig:merge-mixing-weights}
\end{figure}

\subsection{MovieLens Next-Item Prediction}
We use the MovieLens dataset \cite{harper2015movielens} to evaluate whether the proposed merging MERW construction yields a meaningful random walk on real interaction data. { The task is to rank the next movie $m_t$ among 500 candidates, given the preceding two movies $(m_{t-2},m_{t-1})$ in a user's history.} Specifically, we extract merging events from time-ordered user histories, infer a merging transition tensor on an earlier portion of the data via the KL projection with a prescribed stationary distribution, and then evaluate the induced conditional law on a held-out portion of the same user histories. Similar experimental setups appear in \cite{benson2017spacey}. { The node set consists of the 500 movies with the highest rating counts, and we retain users with at least 30 interactions after this filtering. For each user, ratings are sorted by timestamp, with the first $80\%$ used for training and the remaining $20\%$ for testing. A sliding window of length three is then applied, so that each triple $(m_{t-2},m_{t-1},m_t)$ generates one $k=3$ event $({m_{t-2},m_{t-1}}\to m_t)$. Thus, the first two movies form the context and the third is the next-movie target. For example, four consecutive movies $(a,b,c,d)$ generate the events $({a,b}\to c)$ and $({b,c}\to d)$. This procedure yields 30,954 training events and 8,127 test events.}

A directed hyperedge represents a merging event extracted from a user's time-ordered interactions. This construction aligns with the merging interpretation, where an unordered input pair resolves to a single output node. Since the merging kernel is back-symmetric in its tail indices, each context pair is stored in sorted order, and $\mathscr{M}_{j_1j_2j_3}$ is interpreted as the conditional probability of output $j_3$ given the context ${j_1,j_2}$. { The prescribed stationary distribution $\textbf{p}$ is the normalized empirical frequency distribution of the receiver movies $m_t$ in the training events. For each context observed in training, the adjacency support includes the receivers observed for that pair together with up to 200 of the most frequent training receivers associated with either tail movie. The reference entries are then normalized over the receiver mode, while a sparse tail-based support is used for unobserved context pairs.} At this scale, contexts of length $k-1$ have limited repeat coverage for $k\geq4$. We therefore report results for $k=3$, the smallest nontrivial merging model with meaningful out-of-sample coverage.

We evaluate MERW against two baselines. { The popularity baseline ignores context and ranks movies according to the stationary distribution, whereas the Lazy RW baseline uses only the last movie in the context. MERW ranks candidates directly using the conditional distribution $\mathscr{M}^{\star}_{j_1j_2:}$, and all methods rank the same 500 candidate movies.} For each test event, we report the top-$L$ hit rate, defined as the fraction of events for which the realized next movie appears among the top-$L$ candidates under each method. { We consider two evaluation subsets. The seen-context subset contains test events whose context appears in the training data, while the seen-edge subset further requires that the realized next movie belongs to the data-derived adjacency support for that context. Of the 8,127 test events, 3,000 have seen contexts, of which 1,856 satisfy the seen-edge criterion. The seen-context subset therefore provides the broader out-of-sample evaluation, while the seen-edge subset evaluates ranking performance conditional on support coverage.}

Table~\ref{tab:movielens-seenedge} reports results on the seen-edge subset, and Table~\ref{tab:movielens-seenctx} reports results on the seen-context subset. { Across all reported values of $L$, MERW achieves the highest top-$L$ hit rate on both subsets. On the seen-edge subset, MERW improves H@10 from $0.1735$ for Lazy RW to $0.2295$ and H@100 from $0.5603$ to $0.7893$. A similar improvement is observed on the broader seen-context subset, where MERW increases H@10 from $0.1073$ to $0.1420$ and H@100 from $0.3870$ to $0.4883$. Both context-dependent methods substantially outperform the popularity baseline, demonstrating the predictive value of recent interaction history, while the consistent improvement of MERW over Lazy RW indicates the additional benefit of using the joint two-movie context rather than only the most recent movie. The higher hit rates in Table \ref{tab:movielens-seenedge} are expected because the seen-edge evaluation is restricted to events whose realized next movie lies within the data-derived support, whereas Table \ref{tab:movielens-seenctx} also includes seen-context events whose realized next movie falls outside this support, making the prediction task more challenging. Nevertheless, MERW maintains the highest hit rate for every reported value of $L$, demonstrating that its advantage is consistent across both evaluation subsets and across different recommendation-list sizes.}

\begin{table}[t!]
\centering
\caption{ MovieLens ($k=3$), seen-edge test events ($N=1856$). This subset comprises events whose context appears in the learned context map and whose output lies within the data-derived adjacency support.}
\label{tab:movielens-seenedge}
\setlength{\tabcolsep}{5pt}
\renewcommand{\arraystretch}{1.12}
\begin{tabular}{lccccc}
\hline
Method & H@10$\uparrow$ & H@20$\uparrow$ & H@30$\uparrow$ & H@40$\uparrow$ & H@100$\uparrow$ \\
\hline
Popularity  & { 0.0463} & { 0.0717} & { 0.1115} & { 0.1492} & { 0.3330} \\
Lazy RW     & { 0.1735} & { 0.2629} & { 0.3319} & { 0.3917} & { 0.5603} \\

MERW        & { \textbf{0.2295}} & { \textbf{0.3400}} & { \textbf{0.4186}} & { \textbf{0.4844}} & { \textbf{0.7893}} \\
\hline
\end{tabular}
\end{table}
\begin{table}[t!]
\centering
\caption{ MovieLens ($k=3$), seen-context test events ($N=3000$). This subset comprises events whose context appears in the { training data}, regardless of whether the realized output lies within the data-derived adjacency support.}
\label{tab:movielens-seenctx}
\setlength{\tabcolsep}{5pt}
\renewcommand{\arraystretch}{1.12}
\begin{tabular}{lccccc}
\hline
Method & H@10$\uparrow$ & H@20$\uparrow$ & H@30$\uparrow$ & H@40$\uparrow$ & H@100$\uparrow$ \\
\hline
Popularity & { 0.0320} & { 0.0510} & { 0.0800} & { 0.1090} & { 0.2620} \\
Lazy RW    & { 0.1073} & { 0.1627} & { 0.2063} & { 0.2480} & { 0.3870} \\

 MERW       & { \textbf{0.1420}} & { \textbf{0.2103}} & { \textbf{0.2590}} & { \textbf{0.2997}} & { \textbf{0.4883}} \\
\hline
\end{tabular}
\end{table}

\section{Discussion}\label{sec:dis}

In real-world hypergraphs such as metabolic networks \cite{chen2023teasing}, each hyperedge may contain multiple tail and head nodes. This motivates a many-to-many MERW in which a group of input nodes jointly produces a group of output nodes. For simplicity, we consider $k$-uniform hypergraphs. Let $r \ge 1$ and $s \ge 1$ denote the numbers of tail and head nodes, respectively, with $k = r + s$. The transition mechanism is described by a nonnegative order-$k$ tensor
\begin{equation*}
    \mathscr{E}^{(r,s)}\in\mathbb{R}_{+}^{(n \times n \times \stackrel{r}{\cdots} \times n)\times (n \times n \times \stackrel{s}{\cdots} \times n)}
\end{equation*}
whose first $r$ indices correspond to the tail nodes and whose remaining $s$ indices correspond to the head nodes. The entry $\mathscr{E}^{(r,s)}_{j_1j_2 \dots j_r j_{r+1} \dots j_{r+s}}$ specifies the probability associated with the transition from the pivot (tail) tuple $(j_1,j_2,\dots,j_r)$ to the receiver (head) tuple $(j_{r+1},j_{r+2},\dots,j_{r+s})$. To reflect the indistinguishability of nodes within each group, the tensor is assumed to be invariant under permutations of the first $r$ indices and separately invariant under permutations of the last $s$ indices. Stochasticity is imposed fiberwise so that for every fixed pivot indices
\begin{equation*}
\mathscr{E}^{(r,s)} \times_{r+1,r+2,\dots,k}
\{\textbf{1},\textbf{1},\stackrel{s}{\dots},\textbf{1}\}
= \mathscr{I},
\end{equation*}
where $\mathscr{I}$ denotes the all-one tensor. 
Let $\textbf{p}_t \in \mathbb{R}^n_{+}$ with $\textbf{1}^\top \textbf{p}_t = 1$ denote the node marginal at time $t$. One step samples $r$ pivot nodes independently from $\textbf{p}_t$, draws a receiver set according to $\mathscr{E}^{(r,s)}$, and then selects the next node uniformly from the sampled receiver set. The induced node marginal recursion is
\begin{equation*}
\textbf{p}_{t+1}
=
\mathscr{E}^{(r,s)}
\times_{1,2,\dots,r}
\{\textbf{p}_t,\textbf{p}_t,\stackrel{r}{\dots},\textbf{p}_t\}
\times_{r+2,r+3,\dots,k}
\{\textbf{1},\textbf{1},\stackrel{s}{\dots},\textbf{1}\}.
\end{equation*}
A prescribed stationary distribution $\textbf{p}\in\mathbb{R}^n_{++}$ is imposed by substituting $\textbf{p}_t = \textbf{p}$ into the recursion above, which yields linear constraints in $\mathscr{E}^{(r,s)}$. Consequently, maximum-entropy inference remains a convex KL projection on an order-$k$ tensor. When $r = 1$, the recursion becomes linear in $\textbf{p}_t$ and reduces to a broadcasting-type update after projection. When $r \ge 2$, the recursion becomes a homogeneous polynomial map of degree $r$ on the simplex, consistent with the merging mechanism.

\section{Conclusion}\label{sec:conclusion}

In this article, we propose a maximum-entropy random walk (MERW) framework on directed hypergraphs that models transitions directly through directed hyperedges. The transition mechanism is inferred via a KL projection: among all transition tensors supported on the hypergraph, we select the one closest to a reference tensor while enforcing stochasticity and consistency with a prescribed stationary distribution on nodes. Two representative mechanisms were analyzed. In the broadcasting setting, the induced node marginal dynamics reduces to a linear Markov recursion governed by a projected transition kernel, which allows standard analysis of ergodicity and mixing behavior. In contrast, the merging setting produces a nonlinear polynomial map on the probability simplex. In this case, strong ergodicity can be characterized through contraction conditions for polynomial stochastic operators. In both settings, the optimality conditions lead to multiplicative factorizations that admit Sinkhorn–Schr\"odinger type scaling iterations based on tensor contractions, extending entropy projection and scaling techniques beyond the matrix setting.

Several avenues for future work emerge. {First, a deeper theoretical analysis of the nonlinear dynamics in merging and many-to-many interactions could yield sharper conditions for convergence, stability, and mixing rates.} Second, developing scalable algorithms for large-scale hypergraphs, including stochastic, parallel, or distributed tensor scaling methods, would enable application to massive interaction datasets. Third, learning reference tensors or structural priors directly from data could enhance flexibility and predictive power, especially when the hypergraph structure is partially observed or noisy. Finally, extending the framework to new domains, such as biochemical reaction networks, higher-order social interactions, and sequential recommendation systems, offers the opportunity to exploit the full richness of directed hypergraph dynamics in real-world applications.

\section*{Acknowledgment}
Anqi Dong and Can Chen would like to thank Prof. Karl H. Johansson, Prof. Johan Karlsson, and Prof. Tryphon T. Georgiou for their support and encouragement. The authors would also like to thank the anonymous referees for their constructive comments, which significantly improved this article.

\bibliographystyle{unsrt}
\bibliography{references}

% Optional journal biographies are retained separately.
% \input{author_biographies}
\end{document}